\documentclass[11pt,a4paper]{article} %{report} %

\usepackage{natbib}
\usepackage{stmaryrd}
\usepackage{graphicx}
\usepackage{caption}
\usepackage{subcaption}
\usepackage{amssymb}
\usepackage{bbm}
\usepackage{verbatim}
\usepackage{enumerate}
\usepackage[hyperindex]{hyperref}
\usepackage{amsmath, amsfonts, amsthm, amscd}       %formattazione matematica
\usepackage{fancybox}
\usepackage[T1]{fontenc}
\newlength{\querylen}
\theoremstyle{plain}
\newtheorem{defi}{Definition}

\newtheorem{theorem}{Theorem}
\newtheorem{prop}{Proposition}

\newtheorem{ex}{Example}
\newtheorem{lemma}{Lemma}

\newtheorem{?}{Question}

\newcommand{\branch}[4]{
\left\{
	\begin{array}{ll}
		#1  & \mbox{if } #2 \\
		#3 & \mbox{if } #4
	\end{array}
\right.
}

\newcommand{\x}{\textbf{x}}
\newcommand{\y}{\textbf{y}}

\newcommand{\btheta}{\boldsymbol{\theta}}

\usepackage{enumitem}

\newcommand{\bal}[1]{b_{#1}}
\newcommand{\smooth}[1]{c_{#1}}

\newcommand{\Qgs}{Q_{g,\sigma}}

\newcommand{\Id}{\hbox{Id}}
\newcommand{\RR}{\textrm{R}}
\newcommand{\txy}{t_{xy}}
\newcommand{\tyx}{t_{yx}}
\newcommand{\Zgs}{Z_{g,\sigma}}

\newcommand{\tzero}{t_0}
\newcommand{\Ks}{K_{\sigma}}
\newcommand{\sX}{\mathcal{X}}
\newcommand{\sXn}{\mathcal{X}^{(n)}}

\newcommand{\R}{\mathbb{R}}

\newcommand{\N}{\mathbb{N}}

\newcommand{\E}{\mathbb{E}}

\newcommand{\I}{\mathbb{I}}
\newcommand{\1}{\mathbbm{1}}
\newcommand{\var}{\hbox{var}}

\newcommand{\bx}{\textbf{x}}
\newcommand{\X}{\textbf{X}}
\newcommand{\by}{\textbf{y}}

\newcommand{\e}{\textbf{e}}
\newcommand{\Sn}{\mathcal{S}_n}

\usepackage{color}

\makeatletter
\renewcommand\subsection{\@startsection{subsection}{2}{\z@}%
                                     {-3.25ex\@plus -1ex \@minus -.2ex}%
                                     {-1.5ex \@plus .2ex}%
                                     {\normalfont\normalsize\bfseries}}
\renewcommand\subsubsection{\@startsection{subsubsection}{3}{\z@}%
                                     {-3.25ex\@plus -1ex \@minus -.2ex}%
                                     {-1.5ex \@plus .2ex}%
                                     {\normalfont\normalsize\bfseries}}
\makeatother

\usepackage{setspace}
\usepackage[titles]{tocloft}

\begin{document}
%\title{Design of informed proposals for local MCMC} % in discrete spaces}
\title{Informed proposals for local MCMC in discrete spaces}
\author{Giacomo Zanella}
%\end{frontmatter}
\maketitle
\begin{abstract}
%There is a lack of methodological results to guide practitioners in designing efficient Markov chain Monte Carlo (MCMC) algorithms in discrete spaces. For example, it is still unclear how to extend gradient-based MCMC (e.g. Langevin and Hamiltonian schemes) to networks or partitions spaces. This is particularly relevant when fitting Bayesian nonparametric models, which often involve combinatorial and discrete latent parameters. Motivated by this observation, we consider the problem of designing appropriate informed MCMC proposal distributions in discrete spaces. In particular: assuming perfect knowledge of the target measure, what is the optimal Metropolis-Hastings proposal given a fixed set of allowed local moves? Under regularity assumptions on the target, we derive the class of asymptotically optimal proposal distributions, which we call locally-balanced Proposals (BPs). Such proposals are maximal elements, in terms of Peskun ordering, among proposals obtained as pointwise transformations of the target density. This class of proposals includes the Langevin MCMC scheme and can be seen as a generalization of gradient-based methods to discrete frameworks. We discuss asymptotic analysis, applications to discrete frameworks and connections to other schemes (e.g Multiple-Try).

There is a lack of methodological results to design efficient Markov chain Monte Carlo (MCMC) algorithms for statistical models with discrete-valued high-dimensional parameters.
Motivated by this consideration, we propose a simple framework for the design of informed MCMC proposals (i.e.\ Metropolis-Hastings proposal distributions that appropriately incorporate local information about the target) which is naturally applicable to both discrete and continuous spaces.
We explicitly characterize the class of optimal proposal distributions under this framework, which we refer to as \emph{locally-balanced} proposals, and prove their Peskun-optimality in high-dimensional regimes.
%The latter are provably optimal in terms of Peskun ordering in high-dimensional regimes.
%The resulting MCMC schemes can be directly applied to discrete spaces sampling problems 
%Such family of distributions can directly applied to discrete contexts, while
%Simulation studies with both synthetic and real datasets suggest that locally-balanced schemes provide orders of magnitude improvements in efficiency compared to both classical MCMC schemes (Gibbs Sampling and random walk Metropolis) as well as recently proposed discrete relaxations of Hamiltonian Monte Carlo or the Hamming Ball sampler.
The resulting algorithms are straightforward to implement in discrete spaces and provide orders of magnitude improvements in efficiency compared to alternative MCMC schemes, including discrete versions of Hamiltonian Monte Carlo.
Simulations are performed with both simulated and real datasets, including a detailed application to Bayesian record linkage. % is considered.
%The improvement are used to derive efficient samplers for 
%A detailed application to Bayesian record linkage problems is considered.
%The results are used to derive efficient samplers for Bayesian record linkage problems and a real-data application is considered.
A direct connection with gradient-based MCMC suggests that locally-balanced proposals may be seen as a natural way to extend the latter to discrete spaces.

\end{abstract}
%\setcounter{tocdepth}{1}
%%\begin{spacing}{0.3} \tableofcontents \end{spacing}
%%\tableofcontents
%\setlength{\cftbeforesecskip}{-1.8pt}
%\tableofcontents

\section{Introduction}
Markov chain Monte Carlo (MCMC) algorithms are one the most widely used methodologies %routinely used
% to obtain random samples 
to sample from complex and intractable probability distributions, especially in the context of Bayesian statistics \citep{Robert2004}.
Given a distribution of interest $\Pi(dx)$ defined on some measurable space $\sX$, 
%the MCMC method consists in simulating 
MCMC methods simulate
a Markov chain $\{X_t\}_{t=1}^\infty$ having $\Pi$ as stationary distribution and then use the states visited by $X_t$ as Monte Carlo samples from $\Pi$. % to estimate $\E_{\pi}[h]$.
%If the chain is irreducible,
Under mild assumptions, the Ergodic Theorem guarantees that the resulting sample averages are consistent estimators for arbitrary expectations under $\Pi$.
%, for any function $h:\sX\rightarrow \R$, the following convergence occurs almost surely
%\begin{equation}\label{eq:ergodicity}
%\hat{h}_T
%\,=\,
%\frac{\sum_{t=1}^T h(X_t)}{T}
%\quad\stackrel{a.s.}\longrightarrow\quad
%%\Pi(h)%\E_\Pi[h]%
%%\,=\, 
%\int_{\sX}h(x)\Pi(dx)\,
%\quad \hbox{as }T\rightarrow\infty.
%\end{equation}
%\begin{equation*}
%\hat{h}_T
%\,=\,
%\frac{\sum_{t=1}^T h(X_t)}{T}
%\quad\stackrel{a.s.}\longrightarrow\quad
%%\Pi(h)%\E_\Pi[h]%
%%\,=\, 
%\int_{\sX}h(x)\Pi(dx)\,
%\quad\hbox{for any } h:\sX\rightarrow \R \hbox{ as }T\rightarrow\infty.
%\end{equation*}
Many MCMC schemes used in practice fall within the Metropolis-Hastings (MH) framework \citep{Metropolis1953,Hastings1970}.
Given a current state $x\in\sX$, the MH algorithm samples a proposed state $y$ according to some proposal distribution $Q(x,\cdot)$ and then accepts it with probability $a(x,y)=\min\left\{1,\frac{\Pi(dy)Q(y,dx)}{\Pi(dx)Q(x,dy)}\right\}$ or otherwise rejects it and stays at $x$. % (see, e.g., \citet{Tierney1998}).
%The acceptance probability $a(x,y)$ equals $\min\left\{1,\frac{\Pi(dy)Q(y,dx)}{\Pi(dx)Q(x,dy)}\right\}$ for $(x,y)$ in the set $\RR\subseteq\sX\times\sX$ where $\Pi(dx)Q(x,dy)$ and $\Pi(dy)Q(y,dx)$ are mutually absolutely continuous (and thus the Radon-Nikodym derivative $\frac{\Pi(dy)Q(y,dx)}{\Pi(dx)Q(x,dy)}$ is well defined) and 
%0 for $(x,y)\in\RR^c$ (see \citet[Prop.1]{Tierney1998} for more details).
The resulting transition kernel
\begin{equation*}
P(x,dy)=Q(x,dy)+\delta_x(dy)\int_{\sX}(1-a(x,z))Q(x,dz)
\end{equation*}
is $\Pi$-reversible and can be used for MCMC purposes. %to draw samples from $\Pi$.
%%\begin{equation*}%\label{eq:MH}
%%P(x,y)=\left\{
%%\begin{array}{ll}
%%Q(x,y)a(x,y)
%%& \hbox{if }y\neq x,\\
%%1-\sum_{z\neq x} Q(x,z)a(x,z)
%%& \hbox{if }y= x\,.
%%\end{array}
%%\right.
%%\end{equation*}
%Such kernel is $\Pi$-reversible, meaning that $\Pi(dx)P(x,dy)=\Pi(dy)P(y,dx)$ as measures on $\sX\times\sX$,
%%More precisely the MH algorithm produces a $\pi$-reversible kernel $P$, which means a kernel $P$ such that
%%\begin{equation}\label{eq:detailed_balance}
%%\qquad\qquad
%%\pi(x)P(x,y)\;=\;\pi(y)P(y,x)
%%\qquad \forall x,y\in\sX\,,
%%\end{equation}
%which implies that $P$ is $\Pi$-stationary.
%%For a more general and detailed discussion of Monte Carlo methods we refer to \cite{Robert2004}.
%The performances of MH algorithms depend crucially on the
Although the MH algorithm can be applied to virtually any target distribution, its efficiency depends drastically on the proposal distribution $Q$ and its interaction with the target $\Pi$.
Good choices of $Q$ will speed up the Markov chain's convergence while bad choices will slow it down in a potentially dramatic way. %of $\hat{h}_T$ to $\Pi(h)$ 
%Although the MH algorithm can be applied to virtually any target distribution, its efficiency depends drastically on the choice of proposal distribution $Q$.%, as bad choices of $Q$ will produce Markov chains with extremely slow convergence.
%When designing a MH proposal distribution it would be desirable to choose a proposal distribution that matches perfectly the target distribution, ideally setting 
%propose global moves (i.e.\ $Q(x,y)>0$ for any $(x,y)\in\sX\times\sX$), ideally sampling from the target $\pi$ itself, $Q(x,y)=\pi(y)$.
%However, most of the times this is not computationally feasible and the only feasible choice is to implement a MH algorithm performing local moves.

\subsection{Random walk versus informed schemes}
Random walk MH schemes use symmetric proposal distributions satisfying $Q(x,y)=Q(y,x)$, such as normal distributions centered at the current location $Q_\sigma(x,\cdot)=N(x,\sigma^2\I_n)$.
Although these schemes are easy to implement, the new state $y$ is proposed ``blindly'' (i.e.\ using no information about $\Pi$) and this can lead to bad mixing and slow convergence.
In continuous frameworks, such as $\sX=\R^n$ and $\Pi(dx)=\pi(x)dx$, various \emph{informed} MH proposal distributions have been designed to obtain better convergence.
For example the Metropolis-Adjusted Langevin Algorithm (MALA, e.g.\ \citealp{Roberts1998MALA}) exploits the gradient of the target to bias the proposal distribution towards high probability regions by setting $Q_\sigma(x,\cdot)=N(x+\frac{\sigma^2}{2}\nabla(\log \pi)(x),\sigma^2\I_n)$.
Such an algorithm is derived by discretizing the $\Pi$-reversible Langevin diffusion $X_t$ given by
$dX_t=\frac{\sigma^2}{2}\nabla(\log \pi)(x)dt+\sigma dB_t$, so that the proposal $Q_\sigma$ is approximately $\Pi$-reversible %, at least 
for small values of $\sigma$.
%The intuition behind this choice is to use a proposal $Q_\sigma$ that is already approximately $\Pi$-reversible, at least for small values of $\sigma$, and thus will lead to a higher acceptance rate.
%MALA has typically better mixing properties compared to the Markov chain obtained from the random walk proposal $y\sim N(x,\sigma^2\I_n)$ (see e.g.\ \citealp{Roberts1998MALA}), although MALA can sometimes lead to unstable behaviors.
More elaborate gradient-based informed proposals have been devised, such as Hamiltonian Monte Carlo (HMC, e.g.\ \citet{Neal2011,Girolami2011}), and other schemes \citep{Welling2011,Titsias2016,Durmus2017}, resulting in substantial improvements of MCMC performances both in theory and in practice.
%More elaborate gradient-based informed proposals have been proposed, such as Hamiltonian Monte Carlo (e.g.\ \cite{Neal2011}), manifold schemes \citep{Girolami2011}, stochastic-gradient methods \citep{Welling2011} and higher-order schemes \citep{Durmus2015}, resulting in a substantial improvement of MCMC performances, both in theory and in practice.
However, most of these proposal distributions are derived as discretization of continuous-time diffusion processes or measure-preserving flows, and are based on derivatives and Gaussian distributions.
Currently, it is not clear how to appropriately extend such methods to frameworks where $\sX$ is a discrete space.
As a consequence, practitioners using MCMC to target measures on discrete spaces often rely on symmetric/uninformed proposal distributions, which can induce slow convergence.

\subsection{Informed proposals in discrete spaces}
A simple way to circumvent the problem described above is to map discrete spaces to continuous ones and then apply informed schemes in the latter, typically using HMC \citep{Zhang2012,Pakman2013,Nishimura2017}.
Although useful in some scenarios, % (e.g.\ integer-valued parameters)
the main limitation of this approach is that the embedding of discrete spaces into continuous ones is not always feasible and can potentially destroy the natural topological structure of the discrete space under consideration (e.g.\ spaces of trees, partitions, permutations,\dots), thus resulting in highly multimodal and irregular target distributions that are hard to explore.
%A more widely applicable informed MCMC scheme for discrete spaces has been recently proposed in \citep{Titsias2017}. Such scheme, named the Hamming Ball sampler, produces informed proposals respecting the topological structure of $\sX$ by introducing auxiliary variables. % (see Section \ref{sec:simulations} for more details).
An alternative 
%and more widely applicable 
approach was recently proposed in \citep{Titsias2017}, where informed proposals are obtained by introducing auxiliary variables and performing Gibbs Sampling in the augmented space.
The resulting scheme, named the Hamming Ball sampler, requires no continuous space embedding and is directly applicable to generic discrete spaces, but the potentially strong correlation between the auxiliary variables and the chain state can severely slow down convergence.
%In the simulation study we will compare both such approaches to our proposed ones, observing orders of magnitude improvements in efficiency in different applications.
%In Section \ref{sec:simulations} we provide more details and compare both methods with our proposed schemes.
%We will compare our and a comparison with our proposed schemes).
%\subsection{Overview of the paper}

In this work we formulate the problem of designing informed MH proposal distributions in an original way.
Our formulation has the merit of being simple and unifying continuous and discrete frameworks.
The theoretical results hint to a simple and practical class of informed MH proposal distributions that are well designed for high-dimensional discrete problems, which we refer to as \emph{locally-balanced proposals}. Experiments on both simulated and real data show orders of magnitude improvements in efficiency compared to both random walk MH and the alternative informed schemes described above.

\subsection{Paper structure}
%In Section \ref{sec:bal_prop_definition} we define the class of informed proposals distribution considered, which are obtained as a product of some ``base'' uninformed kernel and a biasing multiplicative term.
%We then characterize the class of biasing terms that induce asymptotically exact proposals in the local limit regime (i.e.\ stepsize of the proposal going to 0), which we refer to as \emph{locally-balanced} proposals.
In Section \ref{sec:bal_prop_definition} we define the class of informed proposals distribution considered (obtained as a product of some ``base'' uninformed kernel and a biasing multiplicative term) and we characterize the class of biasing terms that 
are asymptotically-exact %induce asymptotically exact proposals 
in the local limit regime (i.e.\ stepsize of the proposal going to 0).
In Section \ref{sec:peskun_result} we show that, under regularity assumptions on the target, the same class of locally-balanced proposals is also optimal in terms of Peskun ordering as the dimensionality of the state space increases.
In Section \ref{sec:choice_of_g} we consider a simple binary target distribution in order to compare different locally-balanced proposals and identify the one leading to the smallest mixing time, which turns out to be related to the Barker's algorithm \citep{Barker1965}.
Section \ref{sec:MALA} discusses the connection with classical gradient-based MCMC and MALA in particular. %we draw interesting and promising connections with some classical MCMC schemes and
In Section \ref{sec:simulations} we perform simulation studies on classic discrete models (permutation spaces and Ising model), while in Section \ref{sec:record_linkage} we consider a more detailed application to Bayesian Record Linkage problems.
Finally, in Section \ref{sec:discussion} we discuss possible extensions and future works.
Supplementary material includes proofs and additional details on the simulations studies.

\section{Locally-balanced proposals}\label{sec:bal_prop_definition}
Let $\Pi$ be a target probability distribution on some topological space $\sX$.
%Let $\sX$ be a topological space endowed with its Borel $\sigma$-algebra and $\Pi$ the target probability measure on $\sX$. 
We assume that $\Pi$ admits bounded density $\pi$ with respect to some reference measure $dx$, meaning $\Pi(dx)=\pi(x)dx$. Typically $dx$ would be the counting measure if $\sX$ is discrete or the Lebesgue measure if $\sX=\R^n$ for some $n\geq1$.

Let $K_\sigma(x,dy)$ be the uninformed symmetric  kernel that we would use to generate proposals in a random walk MH scheme, such as a Gaussian distribution for continuous spaces or the uniform distribution over neighbouring states for discrete spaces.
Here $\sigma$ is a scale parameter and we assume that $K_\sigma(x,dy)$ converges weakly to the delta measure in $x$ as $\sigma\downarrow 0$ while it converges to the base measure $dy$ as $\sigma\uparrow \infty$.

\subsection{Heuristics: local moves versus global moves}\label{sec:heuristics}
Suppose that we want to modify $K_\sigma(x,dy)$ and incorporate information about $\pi$ in order to bias the proposal towards high-probability states.
The first, somehow naive choice could be to consider the following localized version of $\pi$
\begin{gather}
Q_\pi(x,dy)
=\frac{\pi(y)K_\sigma(x,dy)}{Z_\sigma(x)}\label{eq:Qpi}\,,
%\\
%e.g.\swarrow\hspace{50mm}\searrow e.g.\\
%Q_\pi(x,y)
%=\frac{\pi(y)\1_{N_\sigma(x)}(y)}{\pi(N_\sigma(x))}
%\hspace{40mm}
%Q_\pi(x,y)
%=\frac{\pi(y)e^{-\frac{|x-y|^2}{2\sigma^2}}}{(k_\sigma*\pi)(x)}\vspace{4mm}
\end{gather}
where $Z_\sigma(x)$ is a normalizing constant.
Assuming we could sample from it, we ask whether $Q_\pi$ would be a good choice of proposal distribution.
%As $\sigma$ increases and the localizing effect of $K_\sigma$ decreases, $Q_\pi(x,dy)$ converges to $\pi(dy)$, suggesting that $Q_\pi$ is a good choice of proposal for big values of $\sigma$.
%However this intuition fail for finite and potentially small values of $\sigma$.
Equation \eqref{eq:Qpi} and the symmetry of $K_\sigma$, $K_\sigma(x,dy)dx=K_\sigma(y,dx)dy$, implies that %$K_\sigma(x,dy)dx=K_\sigma(y,dx)dy$ and therefore from 
$$
\frac{Q_\pi(x,dy)dx}{\pi(y)Z_\sigma(y)}=\frac{Q_\pi(y,dx)dy}{\pi(x)Z_\sigma(x)}\,,
%\;\Rightarrow\;
%Q_\pi\hbox{ reversible w.r.t. }\pi(x)Z_\sigma(x)\,,
$$
which means that $Q_\pi$ is reversible with respect to $\pi(x)Z_\sigma(x)dx$.
Note that the normalizing constant $Z_\sigma(x)$ is given by the convolution between $\pi$ and $K_\sigma$ which we denote by $Z_\sigma(x)=(\pi*K_\sigma)(x)=\int_\sX \pi(y)K_{\sigma}(x,dy)$.
Therefore, from the assumptions on $K_\sigma$, we have that $Z_\sigma(x)$ converges to $1$ for $\sigma\uparrow\infty$, while it converges to $\pi(x)$ for $\sigma\downarrow 0$.
It follows that the invariant measure of $Q_\pi$ looks very different in the two opposite limiting regimes because
$$
\pi(x)Z_\sigma(x)\rightarrow\branch{\pi(x)}{\sigma\uparrow\infty\hbox{ (Global moves)}}{\pi(x)^2}{\sigma\downarrow0\;\;\hbox{ (Local moves)}}\,.
$$
Therefore, for big values of $\sigma$, $Q_\pi$ will be approximately $\Pi$-reversible and thus it would be a good proposal distribution for the Metropolis-Hastings algorithm. We thus refer to $Q_\pi$ as \emph{globally-balanced} proposal.
On the contrary, for small values of $\sigma$, $Q_\pi$ would \emph{not} be a good Metropolis-Hastings proposal because its invariant distribution converges to $\pi(x)^2dx$ which is potentially very dissimilar from the target $\pi(x)dx$.
%\footnote{Do we want to discuss here why local regime is important? i.e.\ moving part of paragraph under definition 1 here}
%It follows that $Q_\pi$ would be a good Metropolis-Hastings proposal for big values of $\sigma$, while it would be a very bad one for small values of $\sigma$ because its invariant distribution will be close to $\pi^2$ and thus potentially very dissimilar from $\pi$.

Following the previous arguments it is easy to correct for this behavior and design a proposal which is approximately $\Pi$-reversible for small values of $\sigma$.
In particular one could consider replacing the biasing term $\pi(y)$ in \eqref{eq:Qpi} with some transformation $g(\pi(y))$.
A natural choice would be to consider $\sqrt{\pi}$, which leads to the proposal
$$
Q_{\sqrt{\pi}}(x,dy)
=\frac{\sqrt{\pi(y)}K_\sigma(x,dy)}{(\sqrt{\pi}*K_{\sigma})(x)}.
$$
Arguing as before it is trivial to see that $Q_{\sqrt{\pi}}$ is reversible with respect to $\sqrt{\pi(x)}(\sqrt{\pi}*K_{\sigma})(x)dx$, which converges to $\pi(x)dx$ as $\sigma\downarrow 0$.
We thus refer to $Q_{\sqrt{\pi}}$ as an example of \emph{locally-balanced proposal} with respect to $\pi$.
%In the following sections we provide more general constructions and characterizations of locally-balanced proposals together with results on why they are good Metropolis-Hastings proposal distributions in high dimensions.

\subsection{Definition and characterization of locally-balanced proposals}
In this work we will consider a specific class of informed proposal distributions, which we refer to as \emph{pointwise informed} proposals.
%\footnote{find a better name?}
These proposals have the following structure
\begin{equation}\label{eq:informed_proposal}
\Qgs(x,dy)
\quad=\quad
\frac{g\left(\frac{\pi(y)}{\pi(x)}\right)K_\sigma(x,dy)}{Z_g(x)}
\end{equation}
where $g$ is a continuous function from $[0,\infty)$ to itself  and $Z_g(x)$
is the normalizing constant
\begin{equation}\label{eq:Z}
%Z_g(x)=\sum_{z\in\sX}g\left(\frac{\pi(z)}{\pi(x)}\right)K(x,z)\,.
Z_g(x)=\int_{\sX}g\left(\frac{\pi(z)}{\pi(x)}\right)K_\sigma(x,dz)\,.
\end{equation}
The latter is finite by the continuity of $g$ and boundedness of $\pi$.
%\footnote{Could we relax the boundedness assumption on $\pi$?}
%\footnote{Without loss of generality we can suppose $g(1)=1$ because $\Qgs$ depends from $g$ up to proportionality. [[do we need to say this?]]}%Moreover the integrability of $\Qgs$ (i.e. $\Zgs(x)<\infty$ for any $x\in\sX$) is guaranteed by the continuity of $g$ and boundedness of $\pi$.
Throughout the paper, we assume $g$ to be bounded by some linear function (i.e. $g(t)\leq a+b\,t$ for some $a$ and $b$) to avoid integrability issues (see Appendix \ref{appendix:loc_bal_iff}).

The distribution $\Qgs$ in \eqref{eq:informed_proposal} inherits the topological structure of $K_\sigma$ and incorporates information regarding $\Pi$ through the multiplicative term $g\left(\frac{\pi(y)}{\pi(x)}\right)$.
Although the scheme in \eqref{eq:informed_proposal} is not the only way to design informed MH proposals, it is an interesting framework to consider. %and fairly general 
In particular it includes the uninformed choice $Q(x,y)=K_\sigma(x,y)$ when $g(t)= 1$, and the ``naively informed'' choice $Q(x,y)\propto K_\sigma(x,y)\pi(y)$ when $g(t)=t$.
Given \eqref{eq:informed_proposal}, the question of interest is how to choose the function $g$.
In order to guide us in the choice of $g$ we introduce the notion of locally-balanced kernels.
%The heuristic calculations of Section \ref{sec:bal_prop_definition} suggest that the optimal choice of $g$ should belong to the following class.

%\begin{defi}\label{defi:loc_bal}(Locally-balanced kernels)
%Let $\Pi$ be a probability measure on $\sX$ and $\{Q_\sigma\}_{\sigma>0}$, with $Q_\sigma=\left\{Q_\sigma(x,\cdot)\right\}_{x\in\sX}$, be a family of Markov transition kernels on $\sX$ each with stationary distribution $\Pi_\sigma$. 
%We say that
%\begin{itemize}
%\item $\{Q_\sigma\}_{\sigma>0}$ is \emph{locally-balanced} with respect to $\Pi$ if $\Pi_\sigma\stackrel{\sigma\downarrow 0}\Longrightarrow\Pi$,
%\item $\{Q_\sigma\}_{\sigma>0}$ is \emph{globally balanced} with respect to $\Pi$ if
%$\Pi_\sigma\stackrel{\sigma\rightarrow \infty}\Longrightarrow\Pi$,
%\end{itemize}
%where $\Longrightarrow$ denotes the weak convergence of probability measures.
%\end{defi}
%
%\begin{defi}\label{defi:loc_bal}(Locally-balanced kernels)
%Let $\Pi$ be a probability measure on $\sX$ and $\{Q_\sigma\}_{\sigma>0}$, with $Q_\sigma=\left\{Q_\sigma(x,\cdot)\right\}_{x\in\sX}$, be a family of Markov transition kernels on $\sX$.
%We say that $\{Q_\sigma\}_{\sigma>0}$ is \emph{locally-balanced} with respect to $\Pi$ if each $Q_\sigma$ is reversible with respect to some distribution $\Pi_\sigma$ and $\Pi_\sigma\stackrel{\sigma\downarrow 0}\Longrightarrow\Pi$,% converges weakly to $\Pi$ ;
%where $\Rightarrow$ denotes the weak convergence of probability measures.
%\end{defi}

\begin{defi}\label{defi:loc_bal}(Locally-balanced kernels)
A family of Markov transition kernels $\{Q_\sigma\}_{\sigma>0}$ is \emph{locally-balanced} with respect to a distribution $\Pi$ if each $Q_\sigma$ is reversible with respect to some distribution $\Pi_\sigma$ such that $\Pi_\sigma$ converges weakly to $\Pi$ as $\sigma\downarrow 0$.
\end{defi}

The idea behind using a locally-balanced kernel $Q_\sigma$ as a MH proposal targeting $\Pi$ is that, in a local move regime (i.e. for small $\sigma$), $Q_\sigma$ will be almost $\Pi$-reversible and therefore the Metropolis-Hastings correction will have less job to do (namely correcting for the difference between $\Pi_\sigma$ and $\Pi$).
This would allow for more moves to be accepted and longer moves to be performed, thus improving the algorithm's efficiency.
The reason to consider the local-move regime is that, as the dimensionality of the state space increases the MH moves typically become smaller and smaller with respect to the size of $\sX$.
These heuristic ideas will be confirmed by theoretical results and simulations studies in the following sections.
%We will refer to functions $g$ satisfying \eqref{eq:heuristics_balancing_function} as \emph{balancing functions}.
%Equation \eqref{eq:balancing_funct} means that, if $\log f$ is linear, $Q$ is $\pi$-reversible.
%More generally, for smooth targets, $\log f$ can be locally approximated by its first order Taylor expansion and thus $Q$ would be ``approximately'' $\pi$-reversible (meaning that $a_g(x,y)=1+o(R)$ as $R\downarrow 0$, as for MALA).
%This would allow for more moves to be accepted and longer moves to be performed, thus improving the algorithm's efficiency.
%In the next section, in order to obtain rigorous results, we consider the asymptotic regime.
The following theorem explicitly characterizes which pointwise informed proposal are locally-balanced.
%\begin{theorem}\label{thm:loc_bal_iff}%(Locally-balanced iff g)
%The family of kernels $\{\Qgs\}_{\sigma>0}$ of Definition \ref{defi:informed_prop} is locally-balanced with respect to any $\Pi(dx)=\pi(x)dx$ with bounded and continuous density $\pi$ if and only if %the function $g$ satisfies
%\begin{equation}\label{eq:balancing_equation}
%\qquad\qquad
%g(t)
%\;=\;
%t\,g(1/t)
%\qquad\qquad \forall\, t>0\;.
%\end{equation}
%\end{theorem}
\begin{theorem}\label{thm:loc_bal_iff}
%\footnote{I'm not sure whether the statement is clear enough. I tried to keep it short.}
A pointwise informed proposal $\{\Qgs\}_{\sigma>0}$ is locally-balanced with respect to a general $\Pi$ if and only if
\begin{equation}\label{eq:balancing_equation}
\qquad\qquad
g(t)
\;=\;
t\,g(1/t)
\qquad\qquad \forall\, t>0\;.
\end{equation}
\end{theorem}
Motivated by Theorem \ref{thm:loc_bal_iff} we refer to functions $g$ satisfying \eqref{eq:balancing_equation} as balancing functions.
In the next section we provide some results showing that locally-balanced proposals produce asymptotically more efficient MH algorithms compared to other pointwise informed proposals.

\section{Peskun optimality of locally-balanced proposals}\label{sec:peskun_result}
In this section we use Peskun ordering to compare the efficiency of Metropolis-Hastings (MH) schemes generated by pointwise informed proposals defined in \eqref{eq:informed_proposal}.
Unlike Section \ref{sec:bal_prop_definition}, where we considered the local limit $\sigma\downarrow 0$, we now consider a ``fixed $\sigma$'' scenario, dropping the $\sigma$ subscript and denoting the base kernel and corresponding pointwise informed proposals by $K$ and $Q_g$ respectively.
We focus on discrete spaces, where Peskun-type results are more natural and broadly applicable. 
Thus we assume $\sX$ to be a finite space 
%\footnote{actually it would be enough to say countable.} 
and $dx$ to be the counting measure, meaning that $\pi$ is the probability mass function of $\Pi$ and $K(x,y)$ is a symmetric transition matrix.
%We will show that, under some regularity assumptions on $\pi$, locally-balanced proposals are asymptotically optimal as the dimensionality of $\sX$ converges to infinity.

\subsection{Background on Peskun ordering}
Peskun ordering provides a comparison result for Markov chains convergence properties. 
It measure the efficiency of MCMC algorithms in terms of \emph{asymptotic variance} and \emph{spectral gap}. % associated to the Markov chain under consideration.
The notion of asymptotic variance describes how the correlation among MCMC samples affects the variance of the empirical averages estimators. %  $\hat{h}_T=\frac{\sum_{t=1}^T h(X_t)}{T}$. % in \eqref{eq:ergodicity}. %
Given a $\pi$-stationary transition kernel $P$ and a function $h:\Omega\rightarrow\R$, the asymptotic variance $\var_\pi(h,P)$ is defined as
\begin{equation*}%\label{eq:asymptotic_variance}
\var_\pi(h,P)
\;=\;
\lim_{T\rightarrow\infty}T\,\var\left(\frac{\sum_{t=1}^T h(X_t)}{T}\right)%\hat{h}_T
\;=\;
\lim_{T\rightarrow\infty}T^{-1}\,\var\left(\sum_{t=1}^Th(X_t)\right)
\;,
\end{equation*}
where $X_1$, $X_2$, \dots is a Markov chain with transition kernel $P$ started in stationarity (i.e.\ with $X_1\sim\pi$) .
The smaller $\var_\pi(h,P)$ the more efficient the corresponding MCMC algorithm is in estimating $\E_\pi[h]$.
The spectral gap of a Markov transition kernel $P$ is defined as $Gap(P)=1-\lambda_2$, where $\lambda_2$ is the second largest eigenvalue of $P$, and always satisfy $Gap(P)\geq 0$.
%In general it holds, 
The value of $Gap(P)$ is closely related to the convergence properties of $P$, with values
%Values of $Gap(P)$ 
close to 0 corresponding to slow convergence and values distant from 0 corresponding to fast convergence (see, e.g., \cite[Ch.12-13]{Levin2009} for a review of spectral theory for discrete Markov chains).

\begin{theorem}\label{thm:peskun_constant}
Let $P_1$ and $P_2$ be $\pi$-reversible Markov transition kernels on $\sX$ such that 
$P_1(x,y)\geq c\,P_2(x,y)$ for all $x\neq y$ and a fixed $c>0$.
%for $\Pi$-almost every $x$ 
%\begin{equation}\label{eq:peskun_ordering_c}
%P_1(x,A)\geq c\,P_2(x,A)\qquad \forall A\in\Algebra\hbox{ with }x\notin A\,.%x\in\sX,
%\end{equation}
%for some $c>0$. 
Then it holds
\begin{align}
(a)&\quad
\var_\pi(h,P_1)
\,\leq\,
\frac{\var_\pi(h,P_2)}{c}+\frac{1-c}{c}\var_\pi(h)
\qquad\forall h:\sX\to \R%\in L^2(\Pi)
\;,\nonumber\\%\label{eq:covariance_ordering}
(b)&\quad
\quad\hbox{Gap}(P_1)\,\geq\,c\,\hbox{Gap}(P_2)\,.
\nonumber%\label{eq:gap_ordering}
\end{align}
\end{theorem}
The case $c=1$ of Theorem \ref{thm:peskun_constant} is known as Peskun ordering \citep{Peskun1973,Tierney1998}.
Theorem \ref{thm:peskun_constant} implies that if 
$P_1(x,y)\geq c\,P_2(x,y)$ for all $x\neq y$, then $P_1$ is ``$c$ times more efficient'' than $P_2$ in terms of spectral gap and asymptotic variances (ignoring the $\var_\pi(h)$ term which is typically much smaller than $\var_\pi(h,P_2)$ in non-trivial applications).

\subsection{Peskun comparison between pointwise informed proposals}
%Let $\RR\subseteq\sX\times\sX$ be the set where $\pi(x)K_\sigma(x,y)>0$
%Let $\RR=\{(x,y)\in\sX\times\sX \,:\,\pi(x)K_\sigma(x,y)>0\}$ be the set where
%the two measures $\Pi(dx)K(x,dy)$ and $\Pi(dy)K(y,dx)$ are mutually absolutely continuous (see \citet[Prop.1]{Tierney1998} for more details).
%shows that $\RR$ is unique up to sets of zero mass for both measures under consideration).
%Note that $\RR$ is also the set where $\Pi(dx)Q_g(x,dy)$ and $\Pi(dy)Q_g(y,dx)$ are mutually absolutely continuous.
%Given $Z_g(x)$ as in \eqref{eq:Z}, we define the following constant
In order to state Theorem \ref{thm:peskun_result} below, we define the following constant
\begin{equation}\label{eq:smoothness}
\smooth{g} 
\quad=\quad
\sup%\left\{
_{(x,y)\in\RR}%_{x\in\sX,\;y: k(x,y)>0}
\frac{Z_g(y)}{Z_g(x)}
%\;|\;x\in\sX,\;y\in N(x)\right\}
%\;\geq\;1
\;,
\end{equation}
where $\RR=\{(x,y)\in\sX\times\sX \,:\,\pi(x)K(x,y)>0\}$ and $Z_g(x)$ is defined by \eqref{eq:Z}.
%The supremum in \eqref{eq:smoothness} have to be intended $\Pi(dx)K(x,dy)$-almost everywhere.
\begin{theorem}\label{thm:peskun_result}
Let $g:(0,\infty)\to(0,\infty)$.
Define $\tilde{g}(t)=\min\{g(t),t\,g(1/t)\}$ and let $P_g$ and $P_{\tilde{g}}$ be the MH kernels obtained from the pointwise informed proposals $Q_g$ and $Q_{\tilde{g}}$ defined as in \eqref{eq:informed_proposal}.
%respectively (see \eqref{eq:informed_proposal} for definition).
%For $\Pi$-almost every $x$ it holds%
It holds
\begin{equation}\label{eq:peskun}
P_{\tilde{g}}(x,y)
\quad\geq\quad 
\frac{1}{\smooth{g} \smooth{\tilde{g}}}P_{g}(x,y)\,
\qquad
\forall x\neq y\,.
\end{equation}
%\begin{equation}\label{eq:peskun}
%P_{g}(x,y)
%\quad\leq\quad 
%(\smooth{g} \smooth{\tilde{g}})\,P_{\tilde{g}}(x,y)
%\qquad
%\forall x\neq y\,.%\in \sX\,,
%%P_{g}(x,A)
%%\quad\leq\quad 
%%(\smooth{g} \smooth{\tilde{g}})P_{\tilde{g}}(x,A)
%%\qquad
%%\forall A\not\owns x\,.%\in \sX\,,
%\end{equation}
\end{theorem}

The function $\tilde{g}(t)=\min\{g(t),t\,g(1/t)\}$ satisfies $\tilde{g}(t)=t\,\tilde{g}(1/t)$ by definition.
Therefore Theorems \ref{thm:peskun_constant} and \ref{thm:peskun_result} imply that for any $g:\R_+\rightarrow\R_+$ there is a corresponding balancing function $\tilde{g}$ which leads to a more efficient MH algorithm modulo the factor $\frac{1}{\smooth{g}  \smooth{\tilde{g}}}$.
As we discuss in the next section, in many models of interest $\frac{1}{\smooth{g}  \smooth{\tilde{g}}}$ converges to $1$ as the dimensionality of $\sX$ increases.
%Crucially the factor $\frac{1}{\smooth{g}  \smooth{\tilde{g}}}$, which is smaller or equal than 1 by construction, converges to $1$ as the dimensionality of $\sX$ increases in typical discrete spaces scenarios (see e.g.\ examples in next section).
%The constants $\smooth{g}$ and $\smooth{\tilde{g}}$ are greater or equal than $1$ by construction and, in many practical scenarios (see next section), converge to $1$ as the dimensionality of $\sX$ increases.
When this is true we can deduce that locally-balanced proposals are asymptotically optimal in terms of Peskun ordering.
%\footnote{Add a theorem that $c_g\rightarrow 1$ as $\sigma\downarrow 0$.}

\subsection{High-dimensional regime}
Suppose now that the distribution of interest $\pi^{(n)}$ is indexed by a positive integer $n$ which represents the dimensionality of the underlying state space $\sXn$.
Similarly, also the base kernel $K^{(n)}$ and the constants $c_g^{(n)}$ defined by \eqref{eq:smoothness}
%the base measure, the target density and the base symmetric kernel 
depend on $n$.
In many discrete contexts, as the dimensionality goes to infinity, the size of a single move of $K^{(n)}$ becomes smaller and smaller with respect to the size of $\sX^{(n)}$ and does not change significantly the landscape around the current location.
In those cases we would expect the following to hold
\begin{equation}\label{eq:asympt_smoothness}\tag{A}
c_g^{(n)}\to 1\qquad \hbox{as }n\to \infty
\end{equation}
for every well-behaved $g$ (e.g. bounded on compact subsets of $(0,\infty)$).
When \eqref{eq:asympt_smoothness} holds, the factor $\frac{1}{\smooth{g}  \smooth{\tilde{g}}}$ in the Peskun comparison of Theorem \ref{thm:peskun_result} converges to 1 and locally-balanced proposals are asymptotically optimal.
For example, we expect the sufficient condition \eqref{eq:asympt_smoothness} to hold when the conditional independence graph of the model under consideration has a bounded degree and $K^{(n)}$ updates a fixed number of variables at a time.
We now describe three models involving discrete parameters which will be used as illustrative examples in the following sections and prove \eqref{eq:asympt_smoothness} for all of them.
\begin{ex}[Independent binary components]\label{ex:binary}
%\footnote{Make this independent components in general?}
Consider $\sX^{(n)}=\{0,1\}^n$ and, denoting the elements of $\sX^{(n)}$ as $\bx_{1:n}=(x_1,\dots,x_n)$, the target distribution is
$$
\pi^{(n)}(\bx_{1:n})
\;=\;
\prod_{i=1}^n p_i^{1-x_i}(1-p_i)^{x_i}\,,
$$
where each $p_i$ is a probability value in $(0,1)$.
%Moreover $y^{(n)}=(y^{(n)}_1,\dots,y^{(n)}_n)\in N(x^{(n)})$ if and only if $\sum_{i=1}^n |x_i^{(n)}-y_i^{(n)}|=1$.
The base kernel $K^{(n)}(\bx_{1:n},\cdot)$ is the uniform distribution on the neighborhood $N\big(\bx_{1:n}\big)$ defined as
% is the uni
% proposes a new state $\by_{1:n}$ sampling uniformly at random one component of $\bx_{1:n}$ and flipping it.
%Equivalently $\by_{1:n}$ is chosen uniformly at random from the neighborhood of $\bx_{1:n}$ defined as
\begin{equation*}\label{eq:neighborhood_hypercube}
N\big(\bx_{1:n}\big)
\;=\;
\Big\{
\by_{1:n}=(y_1,\dots,y_n)\,:\,\sum_{i=1}^n |x_i-y_i|=1
\Big\}\,.
\end{equation*}
\end{ex}

\begin{ex}[Weighted permutations]\label{ex:perfect_matchings}
Let
\begin{equation}\label{eq:matchings_target}
\pi^{(n)}(\rho)=\frac{1}{Z}\prod_{i=1}^{n}w_{i\rho(i)}
\qquad \rho\in\Sn\,,
\end{equation}
where $\{w_{ij}\}_{i=i,j}^n$ are positive weights, $Z$ is the normalizing constant $\sum_{\rho\in\Sn}\prod_{i=1}^{n}w_{i\rho(i)}$ and $\Sn$ is the space of permutations of $n$ elements (i.e.\ bijiections from $\{1,\dots,n\}$ to itself).
We consider local moves that pick two indices and switch them.
The induced neighboring structure is $\{N(\rho)\}_{\rho\in\Sn}$ with
\begin{equation}\label{eq:perfect_matchings_neighboring}
N(\rho)
\;=\;
\left\{\rho'\in\Sn\,:\,
\rho'=\rho\circ (i,j) \hbox{ for some }i,j\in\{1,\dots,n\}\hbox{ with } i\neq j
\right\}\,,
\end{equation}
where $\rho'=\rho\circ (i,j)$ is defined by $\rho'(i)=\rho(j)$, $\rho'(j)=\rho(i)$ and
$\rho'(l)=\rho(l)$ for $l\neq i$ and $l\neq j$.
%Note that, for any couple $(i,j)$, the ratio $\frac{\pi(\rho\circ (i,j))}{\pi(\rho)}=\frac{w_{i\rho(j)}w_{j\rho(i)}}{w_{i\rho(i)}w_{j\rho(j)}}$ is easy to evaluate (i.e.\ requires $O(1)$ operations independently of $n$).
\end{ex}
\begin{ex}[Ising model]\label{ex:ising}
Consider the state space $\sX^{(n)}=\{-1,1\}^{V_n}$, where $(V_n,E_n)$ is the $n\times n$ square lattice graph with, for example, periodic boundary conditions.
For each $\textbf{x}=(x_i)_{i\in V_n}$, the target distribution is defined as
%graph with $n^2$ nodes indexed by couples $\{(i,j)\}_{i,j=1}^n$ and edges
\begin{equation}\label{eq:ising_model}
\pi^{(n)}(\textbf{x})=\frac{1}{Z}\exp\left(\sum_{i\in V_n}\alpha_ix_i + \lambda\sum_{(i,j)\in E_n}x_ix_j \right)\,,
\end{equation}
where $\alpha_i\in \R$ are biasing terms representing the propensity of $x_i$ to be positive, $\lambda>0$ is a global interaction term and $Z$ is a normalizing constant.
The neighbouring structure is the one given by single-bit flipping
\begin{equation}\label{eq:ising_neighboring}
N(\textbf{x})
\;=\;
\left\{\textbf{y}\in\sX^{(n)}\,:\,
\sum_{i\in V_n}|x_i-y_i|=2
\right\}\,.
\end{equation}
\end{ex}
Example \ref{ex:binary} is an illustrative toy example that we analyze explicitly in Section \ref{sec:hypercube}. %In fact, since the components are independent, one could easily sample directly from it.
Instead, the target measures in Examples \ref{ex:perfect_matchings} and \ref{ex:ising} are non-trivial distributions to sample from that occur in many applied scenarios (see e.g. Sections \ref{sec:simulations} and \ref{sec:record_linkage}), and MCMC schemes are among the most commonly used approaches to obtain approximate samples from those. 
Such examples will be used for illustrations in Sections \ref{sec:sampling_matchings} and \ref{sec:sampling_ising}.
%The target measure in Example \ref{ex:perfect_matchings} arises in many applied scenarios (see Section \ref{sec:simulations_matchings}) and, despite its easy form, it is a highly challenging distribution to sample from.
%Indeed, MCMC algorithms are one of the most commonly used approach to obtain approximate samples form $\pi(\rho)$ (see e.g.\ \citealp{Jerrum1996}). 
The following proposition, combined with Theorem \ref{thm:peskun_result}, shows that for these examples locally-balanced proposal are asymptotically optimal within the class of pointwise informed proposals.
%\begin{prop}\label{thm:norm_const}
%Suppose $\inf_{i,j\in\N}w_{ij}>0$ and $\sup_{i,j\in\N}w_{ij}<\infty$. Then Example \ref{ex:perfect_matchings} satsifies \eqref{eq:asympt_smoothness}.
%%Let $\{w_{ij}\}_{i,j=1}^\infty$ be positive weights with $\inf_{i,j\in\N}w_{ij}>0$ and $\sup_{i,j\in\N}w_{ij}<\infty$.
%%Let $\pi^{(n)}(\rho)\propto\prod_{i=1}^nw_{i\rho(i)}$ for $\rho\in\sX^{(n)}=\Sn$ and let the neighboring structure $\{N(\rho)\}_{\rho\in\sX^{(n)}}$ be as in \eqref{eq:perfect_matchings_neighboring}.
%%Then \eqref{eq:asympt_smoothness} holds.
%\end{prop}
\begin{prop}\label{thm:norm_const}
The following conditions are sufficient for Examples \ref{ex:binary}, \ref{ex:perfect_matchings} and \ref{ex:ising} for \eqref{eq:asympt_smoothness} to hold for every locally bounded function $g:(0,\infty)\rightarrow(0,\infty)$:
\begin{enumerate}[topsep=0pt,itemsep=-1ex,noitemsep]
\item[]\emph{Example \ref{ex:binary}}: $\inf_{i\in\N}p_i>0$, $\sup_{i\in\N}p_i<1$;
\item[]\emph{Example \ref{ex:perfect_matchings}}: $\inf_{i,j\in\N}w_{ij}>0$ and $\sup_{i,j\in\N}w_{ij}<\infty$;
\item[]\emph{Example \ref{ex:ising}}: $\inf_{i\in\N}\alpha_{i}>-\infty$ and $\sup_{i\in\N}\alpha_{i}<\infty$.
\end{enumerate}
%Suppose $\inf_{i,j\in\N}w_{ij}>0$ and $\sup_{i,j\in\N}w_{ij}<\infty$. Then Example \ref{ex:perfect_matchings} satsifies \eqref{eq:asympt_smoothness}.
\end{prop}
%In general we expect the sufficient condition \eqref{eq:asympt_smoothness} to hold when, as the dimensionality goes to infinity, a single move of $K^{(n)}$ does not change significantly the landscape around the current location. For example, in cases when the graphical structure of the model has a bounded degree and $K^{(n)}$ updates a bounded number of variables at a time.
%We leave general results in this direction to future works (see Section \ref{sec:discussion}).

%%%%%%%%%
%%%%%%%%%
\section{Optimal choice of balancing function}\label{sec:choice_of_g}
In Section \ref{sec:peskun_result} we showed that, under the regularity assumption \eqref{eq:asympt_smoothness}, locally-balanced proposals are asymptotically optimal in terms of Peskun ordering.
% the most efficient among the class of proposals $Q_g$ of \eqref{eq:general_proposal}.
It is thus natural to ask if there
%In this section we consider the following natural question: 
is an optimal proposal among the locally-balanced ones or, equivalently, if there is an optimal balancing function $g$ among the ones satisfying $g(t)=tg(1/t)$ (see Table \ref{table:balancing_function}). %\eqref{eq:balancing_equation}.
Before answering this question, we first draw a connection between the choice of balancing function $g$ and the choice of acceptance probability function in the accept/reject step of the Metropolis-Hastings (MH) algorithm.
\begin{table}[h!]
\centering
\begin{tabular}{c|c|c|c|c}\hline
%Balancing function
 & $g(t)=\sqrt{t}$  &$g(t)=\frac{t}{1+t}$  & $g(t)=1\wedge t$ & $g(t)=1\vee t$\\\hline
%Informed proposal\\
$Q_g(x,y)\propto$ & 
$\sqrt{\pi(y)}K(x,y)$ & 
$\frac{\pi(y)}{\pi(x)+\pi(y)}K(x,y)$ & 
$\left(1\wedge\frac{\pi(y)}{\pi(x)}\right)\hspace{-1mm}K(x,y)$ & 
$\left(1\vee\frac{\pi(y)}{\pi(x)}\right)\hspace{-1mm}K(x,y)$\\\hline
\end{tabular}\caption{Examples of locally-balanced proposals $Q_g$ obtained from different balancing functions $g$.}
\label{table:balancing_function}
\end{table}
\subsection{Connection between balancing functions and acceptance probability functions.}\label{sec:accept_bal_function}
The MH algorithm accepts each proposed state $y$ with some probability $a(x,y)$ which we refer to as acceptance probability function (APF).
%For symmetric proposals, the Metropolis-Hastings APF can be written as $a(x,y)=g(\frac{\pi(y)}{\pi(x)})$ with $g(t)=1\wedge t$.
Denoting the ratio $\frac{\pi(y)Q(y,x)}{\pi(x)Q(x,y)}$ by $t(x,y)$, the Metropolis-Hastings APF can be written as $a(x,y)=g(t(x,y))$ with $g(t)=1\wedge t$.
It is well known that this is not the only possible choice:
as long as
\begin{equation}\label{eq:accept_prob_balance}
\hspace{20mm}
a(x,y)=t(x,y)\,a(y,x)
\qquad%\qquad
%\forall x,y\in\sX\,,
\forall x,y:\, Q(x,y)> 0\,,
\end{equation}
%$a(x,y)=t(x,y)\,a(y,x)$ for every $x$ and $y$ 
the resulting kernel $P$ is $\pi$-reversible and can be used for MCMC purposes.
If we write $a(x,y)$ as $g(t(x,y))$, condition \eqref{eq:accept_prob_balance} translates to
$$
g(t)=tg(1/t)\,.
$$
The latter coincides with the condition for $g$ to be a balancing function, see Theorem \ref{thm:loc_bal_iff}.
Therefore each APF $a(x,y)=g(t(x,y))$ corresponds to a balancing function $g$.
However, the family of balancing functions is broader than the family of APFs because $a(x,y)=g(t(x,y))$ are probabilities and thus need to be bounded by 1, while balancing functions don't.
For example, $g(t)=\sqrt{t}$ or $1\vee t$ %satisfies \eqref{eq:balancing_equation} and thus 
are valid balancing functions but are not upper bounded by 1 and thus they are not a valid APF.

The connection with APFs is interesting because the latter are classical and well studied objects.
In particular it is well-known that, in the context of APFs, the Metropolis-Hastings choice $g(t)=1\wedge t$ is optimal and Peskun-dominates all other choices \citep{Peskun1973,Tierney1998}.
This fact, however, does not translate to the context of balancing functions.
In the latter case the comparison between different $g$'s is more subtle and we expect no choice of balancing function to Peskun-dominate the others in general, not even asymptotically.
In the next section we study a simple scenario and show that, in that case, the optimal balancing function is given by $g(t)=\frac{t}{1+t}$.
Interestingly, the latter choice leads to a natural balancing term $g(\frac{\pi(y)}{\pi(x)})=\frac{\pi(y)}{\pi(x)+\pi(y)}$, which has been previously considered in the context of APFs and is commonly referred to as Barker choice \citep{Barker1965}.

\subsection{The optimal proposal for independent binary variables}\label{sec:hypercube}
In this section we compare the efficiency of different locally-balanced proposals in the independent binary components case of Example \ref{ex:binary}.
It turns out that in this specific case the Barker balancing function $g(t)=\frac{t}{1+t}$ leads to the smallest mixing time. 
%The state space under consideration is $\sX^{(n)}=\{0,1\}^n$, where $n$ is a positive integer.
%Denoting the elements of $\sX^{(n)}$ as $\bx_{1:n}=(x_1,\dots,x_n)$, the target distribution is [[REFER TO EXAMPLE]]
%$$
%\pi^{(n)}(\bx_{1:n})
%\;=\;
%\prod_{i=1}^n p_i^{1-x_i}(1-p_i)^{x_i}\,,
%$$
%where each $p_i$ is a probability value in $(0,1)$.
%%Moreover $y^{(n)}=(y^{(n)}_1,\dots,y^{(n)}_n)\in N(x^{(n)})$ if and only if $\sum_{i=1}^n |x_i^{(n)}-y_i^{(n)}|=1$.
%The base kernel $K^{(n)}(\bx_{1:n},\cdot)$ is the uniform distribution on the neighborhood $N\big(\bx_{1:n}\big)$ defined as
%% is the uni
%% proposes a new state $\by_{1:n}$ sampling uniformly at random one component of $\bx_{1:n}$ and flipping it.
%%Equivalently $\by_{1:n}$ is chosen uniformly at random from the neighborhood of $\bx_{1:n}$ defined as
%\begin{equation*}\label{eq:neighborhood_hypercube}
%N\big(\bx_{1:n}\big)
%\;=\;
%\Big\{
%\by_{1:n}=(y_1,\dots,y_n)\,:\,\sum_{i=1}^n |x_i-y_i|=1
%\Big\}\,.
%\end{equation*}

From Example \ref{ex:binary}, each move from $\bx_{1:n}$ to a neighbouring state $\by_{1:n}\in N(\bx_{1:n})$ is obtained by 
%Therefore, $\by_{1:n}$ can be obtained from $\bx_{1:n}$ either by 
flipping one component of $\bx_{1:n}$, say the $i$-th bit, either from $0$ to $1$ or from $1$ to $0$. 
We denote the former by $\by_{1:n}=\bx_{1:n}+\e^{(i)}_{1:n}$ and the latter by $\by_{1:n}=\bx_{1:n}-\e^{(i)}_{1:n}$.
%For any $g:\R_+\to\R_+$ 
The pointwise informed proposal $Q^{(n)}_g$ defined in \eqref{eq:informed_proposal} can then be written as% given by
\begin{equation}\label{eq:proposal0}
Q^{(n)}_g(\bx_{1:n},\by_{1:n})=\frac{1}{Z_g^{(n)}(\bx_{1:n})}\left\{
\begin{array}{ll}
g(\frac{p_i}{1-p_i})& \hbox{if }\by_{1:n}=\bx_{1:n}+\e^{(i)}_{1:n}\,,\\
%x^{(n)}_{i}=0,\;y^{(n)}_{i}=1\hbox{ and }x^{(n)}_j=y^{(n)}_j\;\hbox{ for }j\neq i,\\
g(\frac{1-p_i}{p_i})& \hbox{if }\by_{1:n}=\bx_{1:n}-\e^{(i)}_{1:n}\,,\\%\hbox{if }x^{(n)}_{i}=1,\;y^{(n)}_{i}=0\hbox{ and }x^{(n)}_j=y^{(n)}_j\;\hbox{ for }j\neq i,\\
0 & \hbox{if }\by_{1:n}\notin N\big(\bx_{1:n}\big)\,.
\end{array}
\right.
\end{equation}
In order to compare the efficiency of $Q^{(n)}_g$ for different choices of $g$ we proceed in two steps.
First we show that, after appropriate time-rescaling, the Metropolis-Hastings chain of interest converges to a tractable continuous time process as $n\rightarrow \infty$ (Theorem \ref{theorem:limit_hypercube}).
Secondly we find which choice of $g$ induces the fastest mixing on the limiting continuous-time process.
Similar asymptotic approaches are well-established in the literature to compare MCMC schemes (see e.g.\ \citealp{Roberts1997}).

To simplify the following discussion we first rewrite $Q^{(n)}_g$ as
\begin{equation}\label{eq:hypercube_proposal_general}
Q^{(n)}_g(\bx_{1:n},\by_{1:n})=\frac{1}{Z_g^{(n)}(\bx_{1:n})}\left\{
\begin{array}{ll}
v_i\,c_i\,(1-p_i)& \hbox{if }\by_{1:n}=\bx_{1:n}+\e^{(i)}_{1:n}\,,\\%x^{(n)}_{i}=0,\;y^{(n)}_{i}=1\hbox{ and }x^{(n)}_j=y^{(n)}_j\;\hbox{ for }j\neq i,\\
v_i\, (1-c_i)\,p_i& \hbox{if }\by_{1:n}=\bx_{1:n}-\e^{(i)}_{1:n}\,,\\%\hbox{if }x^{(n)}_{i}=1,\;y^{(n)}_{i}=0\hbox{ and }x^{(n)}_j=y^{(n)}_j\;\hbox{ for }j\neq i,\\
0 & \hbox{if }\by_{1:n}\notin N\big(\bx_{1:n}\big)\,,
\end{array}
\right.
\end{equation}
where $c_i\in(0,1)$ and $v_i>0$ are the solution of $v_i c_i (1-p_i)=g(\frac{p_i}{1-p_i})$ and $v_i (1-c_i) p_i=g(\frac{p_i}{1-p_i})$.
Given \eqref{eq:hypercube_proposal_general}, finding the optimal $g$ corresponds to finding the optimal values for the two sequences $(v_1,v_2,\dots)$ and $(c_1,c_2,\dots)$.
In the following we assume $\inf_{i\in\N}p_i>0$, $\sup_{i\in\N}p_i<1$ and the existence of $\lim_{n\rightarrow\infty}\frac{\sum_{i=1}^nv_i\,p_i(1-p_i)}{n}>0$.
The latter is a mild assumption to avoid pathological behaviour of the sequence of $p_i$'s and guarantee the existence of a limiting process.% we assume $\lim_{n\rightarrow\infty}\frac{\sum_{i=1}^nv_i\,p_i(1-p_i)}{n}$ to exist. The latter 

Let $\textbf{X}^{(n)}(t)$ be the MH Markov chain with proposal $Q^{(n)}_g$ and target $\pi^{(n)}$.
%on $\sX^{(n)}$ with target measure $\pi^{(n)}$ and proposal $Q^{(n)}_g$ in \eqref{eq:hypercube_proposal_general}.
For any real time $t$ and positive integer $k\leq n$, we define
$$
S^{(n)}_{1:k}(t)=\left(X_1^{(n)}(\lfloor nt\rfloor),\dots,X_k^{(n)}(\lfloor nt\rfloor)\right)\,,
$$
with $\lfloor nt\rfloor$ being the largest integer smaller than $nt$.
Note that $S^{(n)}_{1:k}=(S^{(n)}_{1:k}(t))_{t\geq 1}$ is a continuous-time (non-Markov) stochastic process on $\{0,1\}^k$ describing the first $k$ components of $(\textbf{X}^{(n)}(t))_{t\geq 1}$.

\begin{theorem}\label{theorem:limit_hypercube}
Let $\textbf{X}^{(n)}(1)\sim \pi^{(n)}$ for every $n$. 
For any positive integer $k$, it holds
$$S^{(n)}_{1:k} \stackrel{n\rightarrow\infty}\Longrightarrow S_{1:k},$$
where $\Rightarrow$ denotes weak convergence and $S_{1:k}$ is a continuous-time Markov chain on $\{0,1\}^k$ with jumping rates given by %the matrix of rates (sometimes named $Q$-matrix in the literature) being
\begin{equation}\label{eq:q_matrix}
\normalfont{
A\,(\bx_{1:k},\by_{1:k})=\left\{
\begin{array}{ll}
e_i(\textbf{v},c_i)\cdot (1-p_i)& \hbox{if }\by_{1:k}=\bx_{1:k}+\e^{(i)}_{1:k}\,,\\%x_{1:k}_{i}=0,\;y_{1:k}_{i}=1\hbox{ and }x_{1:k}_j=y_{1:k}_j\;\hbox{ for }j\neq i,\\
e_i(\textbf{v},c_i)\cdot p_i& \hbox{if }\by_{1:k}=\bx_{1:k}-\e^{(i)}_{1:k}\,,\\%\hbox{if }x_{1:k}_{i}=1,\;y_{1:k}_{i}=0\hbox{ and }x_{1:k}_j=y_{1:k}_j\;\hbox{ for }j\neq i,\\
0 & \hbox{if }\by_{1:k}\notin N\big(\bx_{1:k}\big)\hbox{ and }\by_{1:k}\neq \bx_{1:k}\,,
\end{array}
\right.
}
\end{equation}
where
\begin{equation}\label{eq:flipping_rate}
e_i(\textbf{v},c_i)=\frac{1}{\bar{Z}(\textbf{v})}\,v_i\left( (1-c_i)\wedge c_i\right)
\end{equation}
with
$\bar{Z}(\textbf{v})=\lim_{n\rightarrow\infty}\frac{\sum_{i=1}^nv_i\,p_i(1-p_i)}{n}$.
\end{theorem}
We can now use Theorem \ref{theorem:limit_hypercube} and the simple form of the limiting process $S_{1:k}$ to establish what is the asymptotically optimal proposal $Q^{(n)}_g$.
%\footnote{Explicitly remark that we are considering the stationary regime?}
In fact \eqref{eq:q_matrix} implies that, in the limiting process $S_{1:k}$, each bit is flipping independently of the others, with flipping rate of the $i$-th bit being proportional to $e_i(\textbf{v},c_i)$.
Moreover, from \eqref{eq:flipping_rate} we see that the parameter $c_i$ influences only the behaviour of the $i$-th component.
Therefore, each $c_i$ can be independently optimized by maximizing $e_i(\textbf{v},c_i)$, which leads to $c_i=\frac{1}{2}$ for every $i$.
%the asymptotically optimal choice of $c_i$ can be simply derived by maximizing the limiting speed of the $i$-th component
By definition of $c_i$, the condition $c_i=\frac{1}{2}$ corresponds to $g(\frac{p_i}{1-p_i})=\frac{p_i}{1-p_i}g(\frac{1-p_i}{p_i})$. Therefore requiring $c_i=\frac{1}{2}$ for all $i$ corresponds to using a balancing function $g$ satisfying $g(t)=tg(1/t)$.
This is in accordance with the results of Section \ref{sec:peskun_result} and with the intuition that locally-balanced proposal are asymptotically optimal in high-dimensions.

Let us now consider the parameters $(v_1,v_2,...)$ . Given $c_i=\frac{1}{2}$, different choices of $(v_1,v_2,...)$  correspond to different locally-balanced proposals.
From \eqref{eq:flipping_rate} we see that each $v_i$ affects the flipping rate of all components through the normalizing constant $\bar{Z}(\textbf{v})$, making the optimal choice of $v_i$ less trivial.
%On the other hand, from \eqref{eq:flipping_rate} we can see that each $v_i$ is proportional to the rate at which the $i$-th component is flipping in the limiting process $S_{1:k}$, but at the same time affects the other components through the normalizing constant $\bar{Z}(\textbf{v})$.
Intuitively, the parameter $v_i$ represents how much effort we put into updating the $i$-th component, and increasing $v_i$ reduces the effort put into updating other components.
In order to discriminate among various choices of $(v_1,v_2,...)$ we look for the choice that minimizes the mixing time of $S_{1:k}$ for $k$ going to infinity.
Although this is not the only possible criterion to use, it is a reasonable and natural one. 
As we discuss in Appendix \ref{appendix:product_chains}, the latter is achieved by minimizing the mixing time of the slowest bit, which corresponds to choosing $v_i$ constant over $i$. %, meaning $v_i=\bar{v}$ for any $i\in\N$ for some $\bar{v}>0$.
Intuitively, this means that we are sharing the sampling effort equally across components.
%\footnote{Since $Q_g^{(n)}$ is defined up to proportionality the actual value to which we set all $v_i$'s is irrelevant and thus we can set $v_i=1$ for all $i$ without loss of generality.}
It follows that the asymptotically optimal proposal $Q^{(n)}_{g_{opt}}$ is
\begin{equation}\label{eq:hypercube_proposal_optimal}
\normalfont{
Q^{(n)}_{g_{opt}}(\bx_{1:n},\by_{1:n})
\;\propto\;\left\{
\begin{array}{ll}
(1-p_i)& \hbox{if }\by_{1:n}=\bx_{1:n}+\e^{(i)}_{1:n}\,,\\%x^{(n)}_{i}=0,\;y^{(n)}_{i}=1\hbox{ and }x^{(n)}_j=y^{(n)}_j\;\hbox{ for }j\neq i,\\
p_i& \hbox{if }\by_{1:n}=\bx_{1:n}-\e^{(i)}_{1:n}\,,\\%\hbox{if }x^{(n)}_{i}=1,\;y^{(n)}_{i}=0\hbox{ and }x^{(n)}_j=y^{(n)}_j\;\hbox{ for }j\neq i,\\
0 & \hbox{if }\by_{1:n}\notin N\big(\bx_{1:n}\big)\,.
\end{array}
\right.
}
\end{equation}
The latter can be written as 
%$Q_{g_{opt}}^{(n)}(\bx_{1:n},\by_{1:n})\propto\frac{\pi^{(n)}(\by_{1:n})}{\pi^{(n)}(\bx_{1:n})+\pi^{(n)}(\by_{1:n})}K^{(n)}(\bx_{1:n},\by_{1:n})$
$$
Q_{g_{opt}}^{(n)}(\bx_{1:n},\by_{1:n})
\propto
\frac{\pi^{(n)}(\by_{1:n})}{\pi^{(n)}(\bx_{1:n})+\pi^{(n)}(\by_{1:n})}
\1_{N(\bx_{1:n})}(\by_{1:n})
$$
%\begin{equation}
%Q_g^{(n)}(\bx_{1:n},\by_{1:n})\propto\left\{
%\begin{array}{ll}
%\frac{\pi^{(n)}(\by_{1:n})}{\pi^{(n)}(\bx_{1:n})+\pi^{(n)}(\by_{1:n})}& \hbox{if }\by_{1:n}\in N(\bx_{1:n})\,,\\
%0 & \hbox{otherwise}\,,
%\end{array}
%\right.
%\end{equation}
which means that the optimal balancing function is 
$$
g_{opt}(t)=\frac{t}{1+t}\,.
$$

\section{Connection to MALA and gradient-based MCMC}\label{sec:MALA}
In the context of continuous state spaces, such as $\sX=\R^n$, it is typically not feasible to sample efficiently from pointwise informed proposals as defined in \eqref{eq:informed_proposal}. % due to the presence of the target $\pi(z)$ in the proposal.
A natural thing to do in this context is to replace the intractable term in $g(\frac{\pi(y)}{\pi(x)})$, i.e.\ the target $\pi(y)$, with 
%its Taylor expansion about 
 some local approximation around 
the current location $x$.
For example, using a first-order Taylor expansion
 %$\pi(y)\approx \pi(x)exp(\nabla(\log f)(x)\cdot(y-x))$ 
$e^{\log \pi (y)}\approx e^{\log \pi (x)+\nabla\log \pi(x)\cdot(y-x)}$ we obtain a family of first-order informed proposals of the form
\begin{equation}\label{eq:first_order}
\Qgs^{(1)}(x,dy)
\quad\propto\quad
g\left(
e^{\nabla\log \pi(x)\cdot(y-x)}%\exp\left(\nabla(\log f)(x)\cdot(y-x)\right)
\right)K_\sigma(x,dy)\,,
\end{equation}
for $K_\sigma$ symmetric and $g$ satisfying \eqref{eq:balancing_equation}.
%Interestingly, the well known MALA proposal (e.g.\ \citealp{Roberts1998MALA}) is a specific instance of \eqref{eq:first_order} where the balancing function is chosen to be $g(t)=\sqrt{t}$ and the base kernel is chosen to be a gaussian distribution $K_\sigma(x,\cdot)=N(x,\sigma^2)$.
Interestingly, the well known MALA proposal (e.g.\ \citealp{Roberts1998MALA}) can be obtained from \eqref{eq:first_order} by choosing $g(t)=\sqrt{t}$ and a gaussian kernel $K_\sigma(x,\cdot)=N(x,\sigma^2)$.
Therefore we can think at MALA as a specific instance of locally-balanced proposal with first-order Taylor approximation.
%is a specific instance of \eqref{eq:first_order} where the balancing function is chosen to be $g(t)=\sqrt{t}$ and the base kernel is chosen to be a gaussian distribution $K_\sigma(x,\cdot)=N(x,\sigma^2)$.
This simple and natural connection between locally-balanced proposals and classical gradient based schemes hints to many possible extentions of the latter, such as modifying the balancing function $g$ or kernel $K_\sigma$ or considering a different approximation for $\pi(y)$.
The flexibility of the resulting framework could help to increase the robustness and efficiency of gradient-based methods.
Recently, \cite{Titsias2016} considered different but related approaches to generalize gradient-based MCMC schemes, achieving state of the art sampling algorithms for various applications compared to both MALA and HMC.
%Somehow related approaches have been recently explored in \cite{Titsias2016}, achieving state of the art sampling algorithms for various applications.
Given the focus of this paper on discrete spaces, we do not pursue this avenue here, leaving this research lines to future work (see Section \ref{sec:discussion}).

%%%%%
%%%%%
\section{Simulation studies}\label{sec:simulations}
In this section we perform simulation studies using the target distributions of Examples \ref{ex:perfect_matchings} and \ref{ex:ising}.
All computations are performed using the $R$ programming language with code 
freely available upon request.
%available in the online supplementary material.
The aim of the simulation study is two-folded: first comparing informed schemes with random walk ones, and secondly comparing different constructions of informed schemes among themselves.

\subsection{MCMC schemes under consideration}
We compare six schemes: random walk MH (RW), a globally-balanced proposal (GB), two locally-balanced proposals (LB1 and LB2), the Hamming Ball sampler (HB)  proposed in \citep{Titsias2017} and the discrete HMC algorithm (D-HMC) proposed in \citet{Pakman2013}.
%The first four schemes (RW-GB-LB1-LB2) are Metropolis-Hastings algorithms with proposals falling within the framework of pointwise informed proposals of Section \ref{sec:bal_prop_definition}, thus having the form $Q_g(x,y)\propto g\left(\frac{\pi(y)}{\pi(x)}\right)\1_{N(x)}(y)$ for 
%%some $g:\R\rightarrow \R$.
% $g(t)$ equal to $1$, $t$, $\sqrt{t}$ and $\frac{t}{1+t}$ respectively.
The first four schemes (RW-GB-LB1-LB2) are Metropolis-Hastings algorithms with pointwise informed proposals of the form $Q_g(x,y)\propto g\left(\frac{\pi(y)}{\pi(x)}\right)K(x,y)$, with 
%some $g:\R\rightarrow \R$.
 $g(t)$ equal to $1$, $t$, $\sqrt{t}$ and $\frac{t}{1+t}$ respectively.
HB is a data augmentation scheme that, given the current location $x_t$, first samples an auxiliary variable $u\sim K(x_t,\cdot)$ and then samples the new state $x_{t+1}\sim Q_{\pi}(u,\cdot)$, where $Q_{\pi}(u,y)\propto \pi(y)K(u,y)$ is defined as in \eqref{eq:Qpi}. No acceptance-reject step is required as the chain is already $\pi$-reversible, being interpretable as a two stage Gibbs sampler on the extended state space $(x,u)$ with target $\pi(x)K(x,u)$.
To have a fair comparison, all these five schemes use the same base kernel $K$, defined as $K(x,\cdot)=Unif(N(x))$ with neighbouring structures $\{N(x)\}_{x\in\sX}$ defined in Examples \ref{ex:perfect_matchings} and \ref{ex:ising}.
Finally, D-HMC is a sampler specific to binary target distributions (thus applicable to Example \ref{ex:ising} but not to Examples \ref{ex:perfect_matchings}) constructed by first embedding the binary space in a continuous space and then applying HMC in the latter.
For its implementation we followed \citet{Pakman2013}, using a Gaussian distribution for the momentum variables and an integration time equal to $2.5\pi$. % (in this context the integration of the Hamiltonian flow is analytic, see \cite{Pakman2013})
We will be talking about acceptance rates for all schemes, even if HB and D-HMC are not constructed as MH schemes. % accept all moves.
For HB, we define the acceptance rate as the proportion of times that the new state $x_{t+1}$ is different from the previous state $x_t$ (indeed the sampling procedure $u\sim K(x_t,\cdot)$ and $x_{t+1}\sim Q_{\pi}(u,\cdot)$ does often return $x_{t+1}=x_{t}$).
For D-HMC we define the acceptance rate as the proportion of times that a proposal to flip a binary component in the HMC flow is accepted (using the \citet{Pakman2013} terminology, the proportion of times that the particle crosses a potentiall wall rather than bouncing back).
Such definitions will facilitate comparison with MH schemes.

\subsection{Sampling permutations}\label{sec:sampling_matchings}
%Consider the target measure in Example \ref{ex:perfect_matchings}.
Consider the setting of Example \ref{ex:perfect_matchings}, with target density $\pi^{(n)}(\rho)\propto \prod_{i=1}^{n}w_{i\rho(i)}$ defined in \eqref{eq:matchings_target} and base kernel $K(\rho,\cdot)$ being the uniform distribution on the neighborhood $N(\rho)$ defined in \eqref{eq:perfect_matchings_neighboring}.
%Note that, for any couple $(i,j)$, the ratio $\frac{\pi(\rho\circ (i,j))}{\pi(\rho)}=\frac{w_{i\rho(j)}w_{j\rho(i)}}{w_{i\rho(i)}w_{j\rho(j)}}$ is easy to evaluate (i.e.\ requires $O(1)$ operations independently of $n$).
The distribution $\pi^{(n)}$ arises in many applied scenarios (e.g.\ \citealp{Dellaert2003,Oh2009,Zanella2015}) and sampling from the latter is a non-trivial task that is often accomplished with MCMC algorithms (see e.g.\ \citet{Jerrum1996} for related complexity results).
%such as: approximating the permanent of a matrix \citep{Jerrum1996}, Data Association problems \citep{Dellaert2003}, 
%Target Tracking \citep{Oh2009} and
%Cluster Analysis \citep{Zanella2015}.
%Sampling from $\pi^{(n)}$ is challenging and MCMC algorithms are a common approach used for such task (see e.g.\ \citealp{Jerrum1996} for related complexity results [[CTITE MORE RECENT WORK HERE ON COMPLEXITY?]]).

For our simulations we first consider the case of i.i.d.\ weights $\{w_{ij}\}_{i,j=1}^n$ with $\log (w_{ij})\sim N(0,\lambda^2)$.
Here $n$ and $\lambda$ provide control on the dimensionality and the smoothness of the target distribution respectively.
For example, when $\lambda=0$ the target distribution is uniform and the five schemes under consideration (D-HMC not applicable) collapse to the same transition kernel, which is $K(\rho,\cdot)$ itself (modulo HB performing two steps per iteration).
On the other hand, as $\lambda$ increases the difference between RW and informed schemes becomes more prominent (Figure \ref{fig:perfect_matchings_1}).
\begin{figure}[h!]
\includegraphics[width=\linewidth]{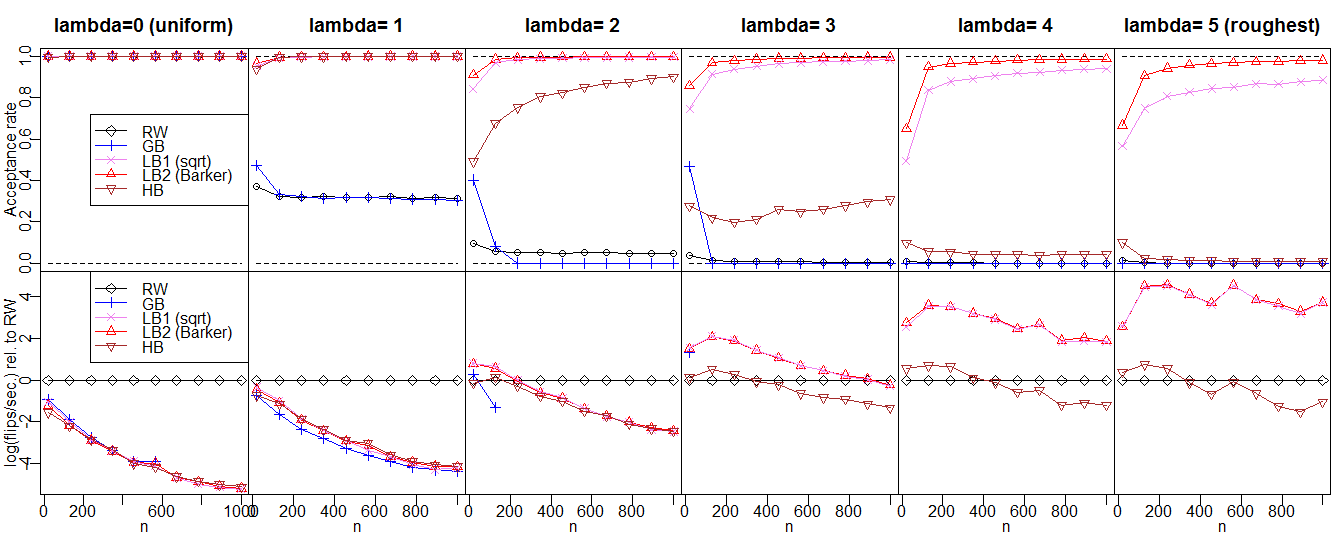}
\caption{Behaviour of the five MCMC schemes under consideration when targeting Example \ref{ex:perfect_matchings} with $\log (w_{ij})\stackrel{iid}\sim N(0,\lambda^2)$ for varying $n$ and $\lambda$.
First row: average acceptance rates. 
Second row: number of flips per unit of computation time relative to RW (on log scale).}\label{fig:perfect_matchings_1}
\end{figure}
%The uninformed scheme RW proposes more and more often unlikely states %with small probabilities under the target 
%and thus its acceptance rate drops close to 0 as $\lambda$ increases (see Figure \ref{fig:perfect_matchings_1}).
In fact, for ``rough'' distributions, most states proposed by RW have small probability under the target and get rejected.
%and thus its acceptance rate drops close to 0 as $\lambda$ increases (see Figure \ref{fig:perfect_matchings_1}, first row).
Despite being more robust than RW, also GB and HB suffer from high rejection rates as $\lambda$ increase.
%For GB, we expect the phenomena to be due to the naive way in which it incorprates the proposal.
On the contrary, the acceptance rates of LB1 and LB2, which are designed to be asymptotically $\pi$-reversible, approaches 1 as $n$ increases for all values of $\lambda$ considered (Figure \ref{fig:perfect_matchings_1}, first row).
In order to take into account the cost per iteration, which is substantially higher for informed schemes, we then compare the number of successful ``flips'' (i.e.\ switches of two indices) per unit of computation time.
The latter is a useful and easy to estimate diagnostic that provides a first gross indication of the relative efficiency of the various schemes.
Figure \ref{fig:perfect_matchings_1} suggests that for flatter targets (i.e.\ small values of $\lambda$) the computational overhead required to use informed proposal schemes is not worth the effort, as a plain RW proposal achieves the highest number of flips per unit time.
However, as the roughness increases and the target distribution becomes more challenging to explore, informed proposals (in particular LB1 and LB2) achieve a much higher number of flips per unit time.
Similar conclusions are obtained by using more advanced diagnostic tools, such as effective sample size per unit time.
For example, Figure \ref{fig:traceplots_iid} displays the results for $n=500$ and $\lambda=5$, suggesting that for rough targets locally-balanced schemes are are one to two orders of magnitude more efficient than the other schemes under consideration (see Supplement \ref{supp:matchings} for results for all values of $\lambda$). 
\begin{figure}[h!]
\includegraphics[width=\linewidth]{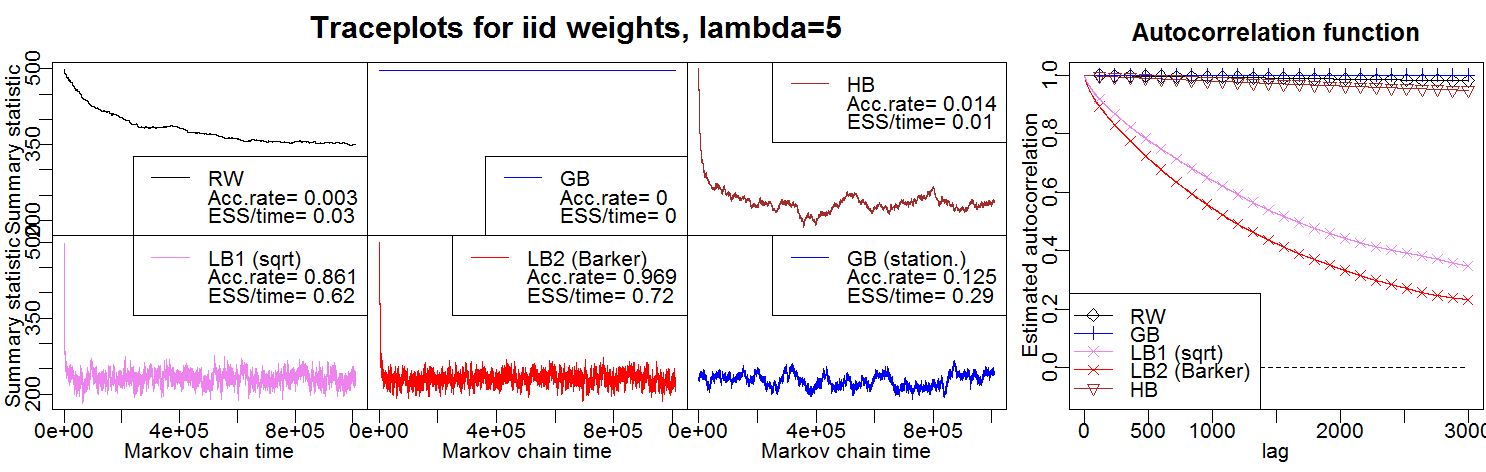}
\caption{Setting: $n=500$ and $\log (w_{ij})\stackrel{iid}\sim N(0,\lambda^2)$ with $\lambda=5$.
Left: traceplots of a summary statistic (Hamming distance from fixed permutation). 
Right: estimated autocorrelation functions.
%Note that the GB chain gets stuck (no move in $10^5$ iterations) if started from a uniformly at random permutation, while it mixes fine in stationarity (i.e. started from a permutation drawn from the target).
}\label{fig:traceplots_iid}
\end{figure}
Note that GB is extremely sensitive to the starting state: if the latter is unlikely under the target (e.g.\ a uniformly at random permutation) the chain gets stuck and reject almost all moves, while if started in stationarity (i.e. from a permutation approx.\ drawn from the target) the chain has a more stable behaviour (see Figure \ref{fig:traceplots_iid}).

Finally, consider a more structured case where, rather than iid weights, we sample $w_{ij}\sim \exp(-\chi^2_{|i-j|})$.
\begin{figure}[h!]
\includegraphics[width=\linewidth]{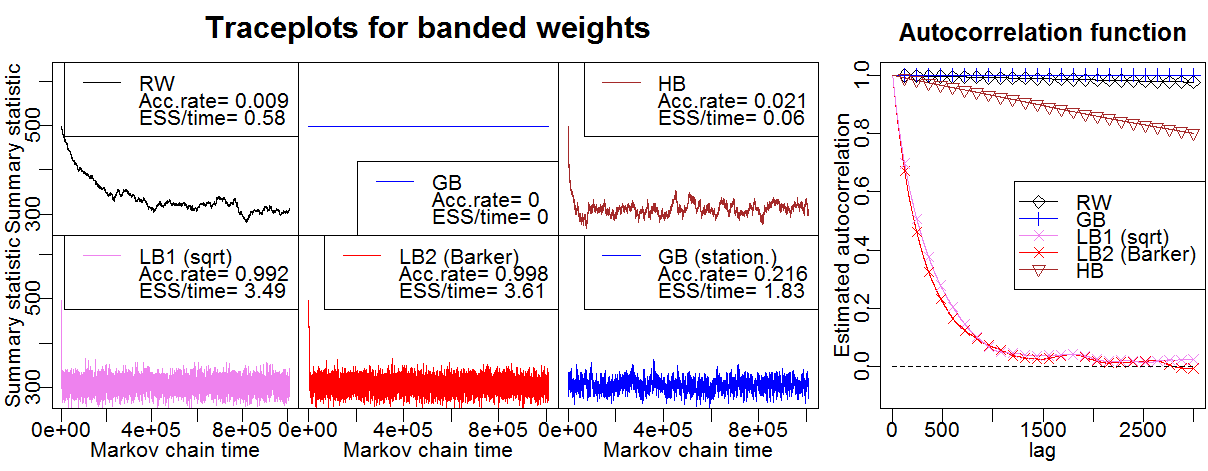}
\caption{Setting: $n=500$ and $w_{ij}\sim \exp(-\chi^2_{|i-j|})$.
Left: traceplots of a summary statistic (Hamming distance from fixed permutation). 
Right: estimated autocorrelation functions.
%Note that the GB chain gets stuck (no move in $10^5$ iterations) if started from a uniformly at random permutation, while it mixes fine in stationarity (i.e. started from a permutation drawn from the target).
}\label{fig:banded}
\end{figure}
The resulting matrix $\{w_{ij}\}_{i,j=1}^n$ has a banded-like structure, with weights getting smaller with distance from the diagonal.
Figure \ref{fig:banded} shows the performances of the different schemes. 
The results are similar to the iid case, with an even more prominent difference in efficiency between HB and LB1-LB2.

\subsection{Ising model}\label{sec:sampling_ising}
%\footnote{Maybe move this first?}
Consider now the Ising model described in Example \ref{ex:ising}. 
The latter is a classic model used in many scientific areas, e.g.\ statistical mechanics and spatial statistics. 
In this simulation study we consider target distributions motivated by Bayesian image analysis, where one seeks to partition an image into objects and background.
In the simplest version of the problem, each pixel $i$ needs to be classified as object ($x_i=1$) or background ($x_i=-1$).
One approach to such task is to define a Bayesian model, using the Ising model (or the Potts model in more general multi-objects contexts) as a prior distribution to induce positive correlation among neighboring pixels.
The resulting posterior distribution is made of a prior term $\exp(\lambda\sum_{(i,j)\in E_n}x_ix_j)$
times a likelihood term $\exp(\sum_{i\in V_n}\alpha_i x_i)$, which combined produce a distribution of the form \eqref{eq:ising_model}.
See Supplement \ref{supp:Ising} for more details on the derivation of such distributions and \cite{Moores2015application} for recent applications to computed tomography.
%Here $\alpha_i=\frac{1}{2}\log\left(\frac{p(y_i|object)}{p(y_i|background)}\right)$.

\begin{figure}[h!]
\includegraphics[width=\linewidth]{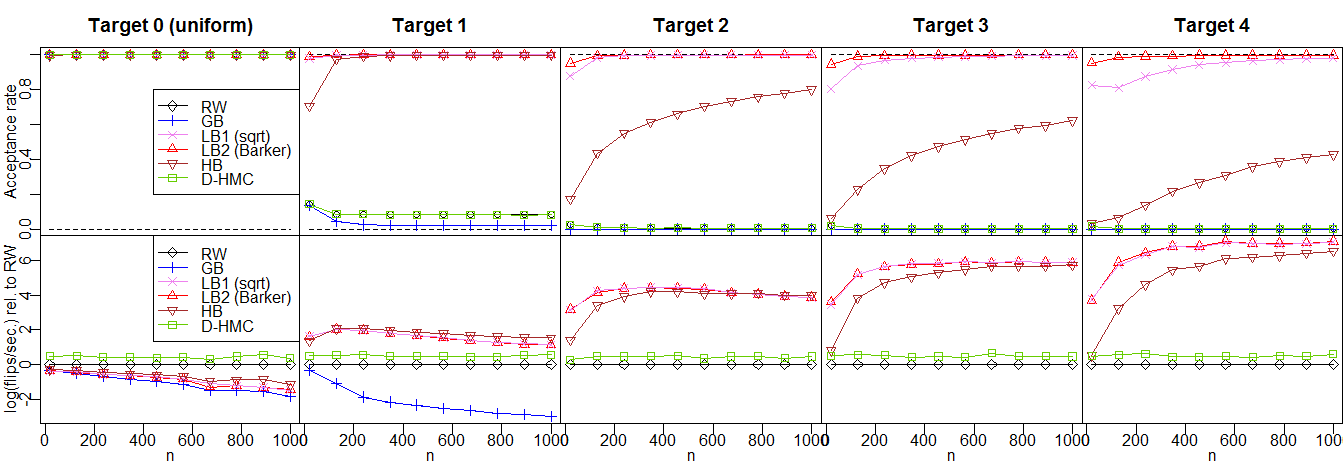}
\caption{
Behaviour of the six MCMC schemes under consideration when targeting Example \ref{ex:ising} for varying $n$ and level of concentration of the target.
First row: average acceptance rates. 
Second row: number of flips per unit of computing time relative to RW (on log scale).
%Average acceptance rates and 
%of the MH algorithm with proposals $Q_{RW}$, $Q^T$ and $Q^B$ for different values of $n$ and $\lambda$.
}\label{fig:ising_acc_rates}
\end{figure}
Similarly to Section \ref{sec:sampling_matchings}, we considered an array of target distributions, varying the size of the $n\times n$ grid and the concentration of the target distribution (controlled by the strength of spatial correlation $\lambda$ and signal-to-noise ratio in the likelihood terms $\alpha_i$).
%, going from the uniform target to more concentrated ones by varying the .
Figure \ref{fig:ising_acc_rates} reports the acceptance rates and number of flips per unit of computing time for the six MCMC schemes under consideration and for five levels of target concentration (see Supplement \ref{supp:Ising} for full details on the set up for $\lambda$ and the $\alpha_i$'s).
RW, GB and D-HMC have very low acceptance rates for all non-uniform distributions considered (see the Appendix of \cite{Pakman2013} for discussion on the similar acceptance rates between RW and D-HMC).
HB, LB1 and LB2 on the contrary do not suffer from high rejection rates and achieve 
a much higher number of %orders of magnitude more 
flips per unit time.
%closer to the one of LB1 and LB2. % and
%In particular RW and GB suffer from extremely high rejection rates in all cases considered apart from the uniform distribution one, while LB1 and LB2 are asymtotically balanced. HB 
However, despite having a good number of flips per second, HB suffers from poor mixing in most cases considered (see e.g.\ Figure \ref{fig:traceplots_Ising_4} and Supplement \ref{supp:Ising} for more examples).
In summary, the results are similar to Section \ref{sec:sampling_matchings}, 
%with the exception of HB suffering less from high rejection rates compared to before.
%As a result, 
with locally balanced schemes (LB1 and LB2) being orders of magnitude more efficient than alternative schemes especially when targeting highly non-uniform targets (see effective sample sizes per unit time in Figure \ref{fig:traceplots_Ising_4} and Supplement \ref{supp:Ising}).
%Despite having a similar number of successful flips per unit time, LB1 and LB2 mix substantially better than HB in this context (see e.g.\ Figure \ref{fig:traceplots_Ising_4} and Supplement \ref{supp:Ising} for more examples).
\begin{figure}[h!]
\includegraphics[width=\linewidth]{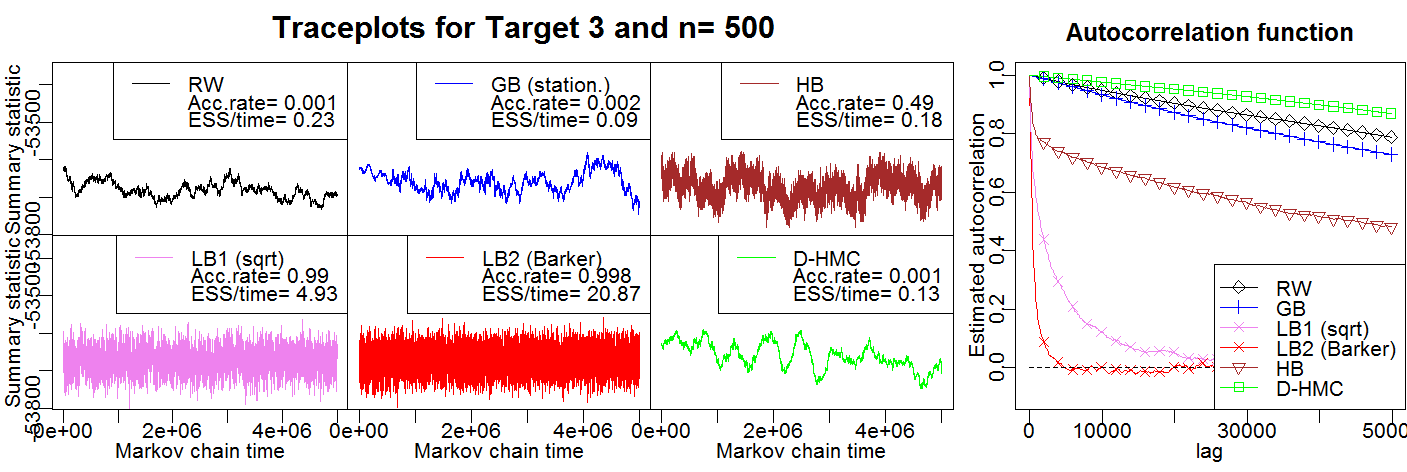}
\caption{Setting: Ising model with $n=500$, $\lambda=1$ and $\alpha_i$'s described in Supplement \ref{supp:Ising}.
Left: traceplots of a summary statistic (Hamming distance from fixed permutation). 
For D-HMC the plot displays the whole trajectory, including the path during the integration period.
Right: estimated autocorrelation functions.%Average acceptance rates of the MH algorithm with proposals $Q_{RW}$, $Q^T$ and $Q^B$ for different values of $n$ and $\lambda$.
}\label{fig:traceplots_Ising_4}
\end{figure}

Given the focus of this work, we do not considered specialized algorithms for the Ising model performing global updates (e.g.\ \cite{Swendsen1987}). %Swendsen-Wang schemes 
%The latter are well-known to perform substantially better than single-site updating schemes when the target exhibits phase-transition behavior (e.g.\ models with no likelihood terms and interaction strength $\lambda$ relatively close to the model's critical temperature). On  the contrary, in cases where the likelihood terms $\{\alpha_i\}_{i\in V_n}$ dominate, like the image analysis context considered here, single site updates schemes are typically more efficient \citep{Hurn1997} and often preferred in practice (see e.g.\ \citealp{Moores2015}).
The latter are somehow complementary to the single-site updating schemes considered here as they perform much better when the target is multimodal (e.g.\ in the absence of informative likelihood terms and with moderately strong interaction term $\lambda$), while they perform poorly in cases where the likelihood terms $\{\alpha_i\}_{i\in V_n}$ dominate, like the image analysis context considered here (see see e.g.\ \citet{Hurn1997} and \citet{Moores2015} for more discussion).

\subsection{Block-wise implementation}\label{sec:blocking}
In many scenarios, as the dimensionality of the state space increases, also the computational cost of sampling from the pointwise informed proposals as defined in \eqref{eq:informed_proposal} increases.
For example, in discrete space settings it is typical to have a base kernel $K(x,\cdot)$ which is a uniform distribution on some neighborhood $N(x)\subseteq\sX$ whose size grows with the dimensionality of $\sX$.
In these cases, when the size of $N(x)$ becomes too large, it may be inefficient to use an informed proposal on the whole neighborhood $N(x)$.
Rather, it may be more efficient to first select a sub-neighborhood $N'(x)\subseteq N(x)$ at random, and then apply locally-balanced proposals to the selected subspace. % subset of the neighborhood $N'(x)\subseteq N(x)$.
For example, if the state space under consideration has a Cartesian product structure, then % $\sX^{(n)}=\times_{i=1}^n\sX_i$ 
one could update a subset of variables given the others in a \emph{block-wise} fashion, analogously to the Gibbs Sampling context.
% is analogous but still very much different to what is typically done in the context of Gibbs sampling.
By choosing appropriately the size of $N'(x)$ one can obtain an appropriate trade-off between computational cost (a small $N'(x)$ induces an informed proposal that is cheap to sample) and statistical efficiency (a large $N'(x)$ produces better infomed moves as it considers more candidate states at each step).
Such an approach is illustrated in Section \ref{sec:record_linkage} and Supplement \ref{appendix:blocking} on a record linkage application.
See also \cite{Titsias2017} for additional discussion on block-wise implementations and the resulting cost-vs-efficiency tradeoff in the context of the Hamming Ball sampler and Section \ref{sec:discussion} for related comments in future works discussion.

\section{Application to Bayesian record linkage}\label{sec:record_linkage}
Record linkage, also known as entity resolution, is the process of identifying which records refer to the same entity across two or more databases with potentially repeated entries.
%Record linkage, also known as entity resolution, is the process of identifying records across databases th
%Given two or more databases with potentially repeated entries, the record linkage process consists in identifying which records refer to the same entity.
Such operation is typically performed to remove duplicates when merging different databases.
If records are potentially noisy and unique identifiers are not available, statistical methods are needed to perform reliable record linkage operations.
While traditional record linkage methodologies are based on the early work \citet{Fellegi1969}, Bayesian approaches to record linkage are receiving increasing attention in recent years \citep{Tancredi2011,SteortsHallFienberg2016,Sadinle2017,Johndrow2017}. 
Such approaches are particularly interesting, for example, as they provide uncertainty statements on the linkage procedure that can be naturally propagated to subsequent inferences, such as population-size estimation \citep{Tancredi2011}.
Despite recent advances in Bayesian modeling for record linkage problems (see e.g.\ \citet{Zanella2016NIPS}), exploring the posterior distribution over the space of possible linkage configurations is still a computationally challenging task which is limiting the applicability of Bayesian record linkage methodologies.
In this section we use locally-balanced proposals to derive efficient samplers for Bayesian record linkage.

We consider bipartite record linkage tasks, where one seeks to merge two databases with duplicates occurring across databases but not within.
This is the most commonly considered case in the record linkage literature (see \cite{Sadinle2017} and references therein).
Denote by $\bx=(x_1,\dots,x_{n_1})$ and $\by=(y_1,\dots,y_{n_2})$ the two databases under consideration.
In this context, the parameter of interest is a partial matching between $\{1,\dots,n_1\}$ and $\{1,\dots,n_2\}$, where $i$ is matched with $j$ if and only if $x_i$ and $y_j$  represent the same entity.
We represent such a matching with a $n_1$-dimensional vector $\textbf{M}=(M_1,\dots,M_{n_1})$, where $M_i=j$ if $x_i$ is matched with $y_j$ and $M_i=0$ if $x_i$ is not matched with any $y_j$.
In Supplement \ref{appendix:record_linkage} we specify a Bayesian model for bipartite record linkage, assuming the number of entities and the number of duplicates to follow Poisson distributions a priori, and assuming the joint distribution of $(\bx,\by)$ given $\textbf{M}$ to follow a spike-and-slab categorical distribution (often called hit-miss model in the Record Linkage literature \citep{CopasHilton1990}).
The unknown parameters of the model are the partial matching $\textbf{M}$ and two real-valued hyperparameters $\lambda$ and $p_{match}$, representing the expected number of entities and the probability of having a duplicate for each entity. 
%See Supplement \ref{appendix:record_linkage} for full details on the model specification.
%We perform inferences with a Metropolis-within-Gibbs framework, where we alternate sampling from $(p_{match},\lambda)|\textbf{M}$ with  Metropolis updates of $\textbf{M}|(p_{match},\lambda)$.
We perform inferences in a Metropolis-within-Gibbs fashion, where we alternate sampling $(p_{match},\lambda)|\textbf{M}$ and  $\textbf{M}|(p_{match},\lambda)$, see Supplement \ref{appendix:Metropolis_within_Gibbs} for full details on the sampler.
While it is straightforward to sample from the two-dimensional distribution $(p_{match},\lambda)|\textbf{M}$, see \eqref{eq:full_cond_p}-\eqref{eq:full_cond_lambda} in Supplement \ref{appendix:record_linkage} for explicit full conditionals, providing an efficient way to update the high-dimensional discrete object $\textbf{M}$ is more challenging and this is where we exploit locally-balanced proposals.
%\footnote{Here we should probably put more details!}
%While it is straightforward to sample from the two-dimensional distribution $(p_{match},\lambda)|\textbf{M}$, see \eqref{eq:full_cond_p}-\eqref{eq:full_cond_lambda} in Supplement \ref{appendix:record_linkage} for explicit full conditionals, providing an efficient way to update the high-dimensional discrete object $\textbf{M}$ is more challenging.

We consider a dataset derived from the Italian Survey on Household and Wealth, which is a biennial survey conducted by the Bank of Italy.
The dataset is publicy available (e.g. through the $Italy$ R package) and consists of  two databases, the 2008 survey (covering 13,702 individuals) and the 2010 one (covering 13,733 individuals).
For each individual, the survey recorded 11 variables, such as year of birth, sex and working status.
%This dataset has been used as a benchmark in recent works on Bayesian record linkage \cite{?}, also because unique identifiers are available and can be used to test model performances.

First, following \cite{Steorts2015EB}, we perform record linkage for each Italian region separately.
This results in 20 separate record linkage tasks with roughly 1300 individuals each on average.
We ran four MCMC schemes for each record linkage task and compare their performances.
Following the notation of Section \ref{sec:simulations}, the four schemes are RW, GB, LB and HB, where LB refers to the locally-balanced proposals with Barker weights.
Figure \ref{fig:region_1}(a) shows the number of matches over MCMC iterations for region 1 (other regions show a similar qualitative behaviour).
We can see that LB and HB converge rapidly to the region of substantial posterior probability, while RW and GB exhibit an extremely slow convergence.
HB, however, converges and mixes significantly slower than LB (see e.g. the autocorrelation functions in \ref{fig:region_1}(b)).
\begin{figure}[h!]
\includegraphics[width=\linewidth]{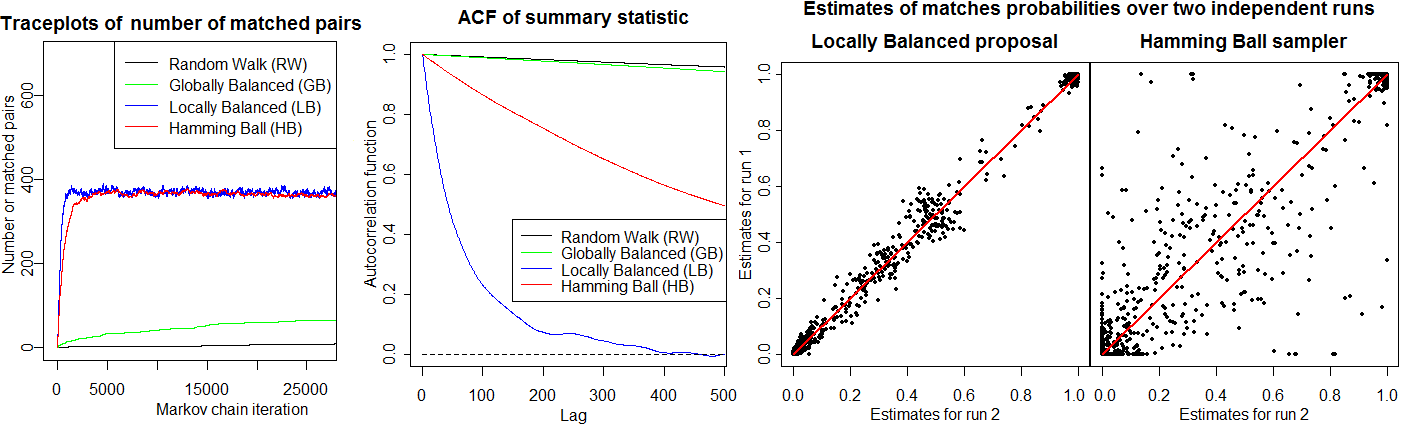}
\caption{Output analysis for the 4 MCMC schemes under comparison applied to the Italy dataset (region 1).}\label{fig:region_1}
\end{figure}
To provide a quantitative comparison between the performances of the four schemes, we consider as efficiency measure the 
effective sample sizes per unit of computation time relative to RW.
Effective sample sizes are computed using the coda R package \citep{Plummer2006} and averaged over 5 summary statistics (each summary statistics being the Hamming distance from a matching randomly drawn from the posterior).
%\footnote{Discuss ESS estimation even before}
From Table \ref{table:region_ess} we can see that LB provides roughly two orders of magnitude improvement in efficiency over RW and GB and one order of magnitude improvement over HB.
\begin{table}[h!]
\centering
\begin{tabular}{l|cccccccc|r}
Region &1&2& 3  &  4 & 5 & 6  & 7 & 8 &  Average \\\hline
GB vs RW &0.8&2.9&0.7&1.7&1.4&0.5&1.0&0.9&0.94\\\hline
LB vs RW &60.6&209%.5
&36.2&127%.5
&48.9&104%.3
&172%.5
&33.3&94.0\\\hline
HB vs RW &6.3&20.3&3.7&15.3&7.2&10.9&15.6&3.1&9.96\\
%Efficiency of HB vs RW & 1 &2  &3  &  \\\hline
%Efficiency of HB vs RW &  4& 5 & 6 &  \\\hline
%Efficiency of HB vs RW &  &  &  & 
\end{tabular}
\caption{Relative efficiency (defined as ESS/time) of GB, LB and HB compared to RW . %of globally-balanced proposals (GB), locally-balanced proposals (LB) and Hamming Ball sampler (HB) compared to random walk Metropolis (RW).+
The table reports the value for the first 8 regions and the average over all 20 regions.%The average is computed over 20 regions.
}
\label{table:region_ess}
\end{table}
Indeed from Figure \ref{fig:region_1}(c) we can see that, given the same computational budget, LB manages to provide much more reliable estimates of the posterior probabilities of each couple being matched compared to HB.
Computations were performed using the $R$ programming language.
For each region we ran LB for $35000$ iterations (enough to produce reliable estimates of posterior pairwise matching probabilities like the one in Figure \ref{fig:region_1}(c)), requiring on average around 120 seconds per region.
%\footnote{Check if we have nested loop}
%[[Do we need acceptance rate table??]]
%\begin{table}[h!]
%\centering
%\caption{Acceptance rates of RW, GB, LB and HB (see Section \ref{} for definition).
%The table reports the value for the first 8 regions and the average over all 20 regions.
%}
%\label{my-label}
%\begin{tabular}{l|cccccccc|r}
%Region &1&2& 3  &  4 & 5 & 6  & 7 & 8 &  Average \\\hline
%RW &<0.1$\%$& 0.3$\%$& <0.1$\%$& 0.1$\%$& <0.1$\%$& 0.1$\%$& 0.1$\%$& <0.1$\%$&0.1$\%$\\\hline
%GB &69.2$\%$&72.3$\%$&69.5$\%$&60.7$\%$&71.1$\%$&60.5$\%$&62.4$\%$&68.1$\%$&65.4$\%$\\\hline
%LB &99.7$\%$& 94.9$\%$& 99.8$\%$& 98.8$\%$& 99.8$\%$& 98.9$\%$& 99.1$\%$& 99.8$\%$&  99.2$\%$\\\hline
%HB &
%10.8$\%$&7.7$\%$&10.8$\%$&12.8$\%$&13.6$\%$&11.2$\%$&10.1$\%$&11.8$\%$&10.9$\%$
%\end{tabular}
%\end{table}

Next, we consider the task of performing recording linkage on the whole Italy dataset, without dividing it first into regions.
In fact the latter operation (typically referred to as deterministic blocking in the record linkage literature) is typically done for computational reasons but has the drawback of excluding a priori possible matches across groups (in this case regions).
%In fact, using the available unique identifiers for this dataset, we can see that splitting the Italy dataset by regions excludes a priori $93\%$ of the true matchings.
We apply LB to the whole dataset using the block-wise implementation discussed in Section \ref{sec:blocking} (details in Supplement \ref{appendix:blocking}). %, as the un-blocked procedure is unfeasible for a dataset of size over 25000 individuals.
%Further details of the blocking implementation in .
Standard output analysis suggests that the LB chain is converging fast (Figure \ref{fig:record_linkage_all}, left) and mixing well in the region of substantial probability (Figure \ref{fig:record_linkage_all}, center).
\begin{figure}[h!]
\includegraphics[width=\linewidth]{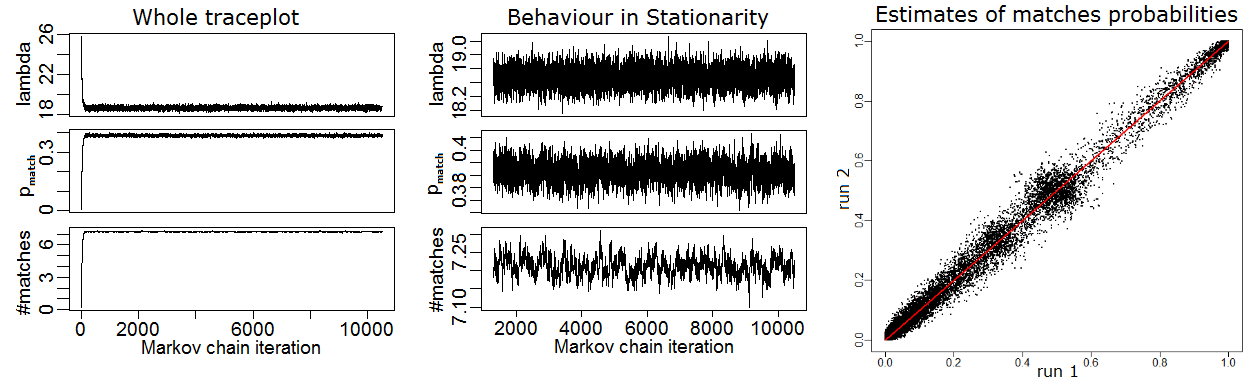}
\caption{Output analysis for locally-balanced MCMC applied to the full Italy dataset (see Section \ref{sec:record_linkage}).}\label{fig:record_linkage_all}
\end{figure}
Comparing independent runs of the algorithms suggests that we are obtaining reliable estimates of the nearly one billion posterior pairwise matching probabilities (Figure \ref{fig:record_linkage_all}, right).
The simulation needed to obtain the probability estimates in Figure \ref{fig:record_linkage_all}(right) took less than 40 minutes with a plain $R$ implementation %(no C++ subroutines) 
on a single desktop computer.

\section{Discussion and future work}\label{sec:discussion}
In this work we discussed a fundamental and yet not satisfactorily answered question in the MCMC methodology literature, which is how to design ``informed'' Metropolis-Hastings proposals in discrete spaces.
We proposed a simple and original framework (pointwise informed proposals and locally-balanced criteria) that provides useful and easy-to-implement methodological guidance. % on the design of informed proposals. % through the use of well-designed biasing terms of the form $g(\frac{\pi(y)}{\pi(x)})$. 
%\footnote{Mention explicitly biasing terms $g(\frac{\pi(y)}{\pi(x)})$?}
%An important merit and motivation of our formulation is being applicable to discrete spaces. %, as most previous work in the area (e.g.\ gradient-based schemes) is restricted to continuous spaces.
Under regularity assumptions, we were able to prove the optimality of locally-balanced proposals in high-dimensional regimes %using Peskun-ordering arguments 
and to identify the optimal elements within such class.
The theoretical results of Sections \ref{sec:bal_prop_definition}-\ref{sec:choice_of_g} are confirmed by simulations (Sections \ref{sec:simulations}-\ref{sec:record_linkage}), where we observe orders of magnitude improvements in efficiency over alternative schemes both for simulated and real data scenarios.
%\footnote{Discuss informed-vs-noninformed and flat-vs-rough?}
%In Section \ref{sec:connections} we discuss connections with classical schemes. %namely MALA and Multiple-try Metropolis.
%In particular, 
%The connection with Multiple-try schemes (Section \ref{sec:MTM}) shows that the problem of designing informed proposals in high dimensions is non-trivial and common intuition may be misleading (see e.g.\ the widespread use of globally-balanced weighting functions $w(x,y)\propto\pi(y)$). 
%We see our contribution firstly as a source of methodological guidance.
%In particular 
We envisage that locally-balanced proposals could be helpful in various contexts to be incorporated as a building block in more elaborate Monte Carlo algorithms.
The proposed methodology can be applied to arbitrary statistical models with discrete-valued parameters, such as Bayesian Nonparametric models or model selection problems.
% does not require specific assumptions on the statistical model under consideration 
% and it would be interesting to consider other applications, such as Bayesian Nonparametric models or model selection problems.
%Also, it would be interesting to consider %, where discrete spaces such as partitions or 

The present work offers many directions for possible extentions and future work.
%The connection to MALA in Section \ref{sec:MALA} could be exploited to design gradient-based MCMC schemes that are more robust to target  improve the robustness of  and make them less sensitive to tuning. %, both in terms of tail behaviour and heterogeneity of the target. %, both in terms of tail behaviour an
For example, the connection to MALA in Section \ref{sec:MALA} could be exploited to improve the robustness of gradient-based MCMC schemes and reduce the notoriously heavy burden associated with their tuning procedures.
Also, it would be interesting to extend our study to the popular context of Multiple-Try Metropolis (MTM) schemes \citep{Liu2000}.
In fact, the MTM weight function, used to select the proposed point among candidate ones, plays a role that is very similar to one of multiplicative biasing terms in pointwise informed proposals and we expect our results to translate quite naturally to that context.

In terms of implementation, it would be interesting to explore in more depth the trade-off between computational cost per iteration and statistical efficiency of the resulting Markov chain.
Beyond the use of block-wise implementations discussed in Section \ref{sec:blocking}, another approach to reduce the cost per iteration would be to replace the informed term $g(\frac{\pi(y)}{\pi(x)})$ in the proposal with some cheap-to-evaluate approximation, while still using the exact target in the MH accept/reject step. % preserving the correct invariance with the MH accept/reject step.
%Also, the computations required to sample from locally-balanced proposals are trivially parallelizable.
%By using specific hardware for parallel computations, such as Graphics Processing Units (GPUs), one could significantly reduce the computational overhead required by using informed proposals without affecting the statistical efficiency.
Also, the computations required to sample from locally-balanced proposals are trivially parallelizable and specific hardware for parallel computations, such as Graphics Processing Units (GPUs), could be used to reduce the computational overhead required by using informed proposals \citep{lee2010utility}.

%In Section \ref{sec:blocking} we briefly discussed a simple way to control such trade-off by applying locally-balanced schemes in a  fashion.
%Note that both such approaches could reduce statistical efficiency but do not introduce any bias in the procedure.
%Beyond the two connections mentioned in Sections \ref{sec:MALA} and \ref{sec:MTM}
%Especially in the context of Multiple-Try schemes, it would be interesting to provide guidance for a regime which is intermediate between local and global.

From the theoretical point of view, it would be interesting to provide guidance for a regime which is intermediate between local and global, maybe by designing appropriate interpolations between locally-balanced and globally-balanced proposals.
This could be useful to design schemes that adaptively learn the appropriate level of interpolation needed.
Also, one could prove the sufficient condition for optimality \eqref{eq:asympt_smoothness} in much more general contexts.
Finally, throughout the paper we assumed the base kernel $K_\sigma$ to be symmetric and it would be interesting to extend the results to the case of general base kernels. % genera that the theory could be appropriately modified to accomodate for general base kernels. 

\subsection*{Acknowledgements}
The author is grateful to Gareth Roberts, Omiros Papaspiliopoulos and Samuel Livingstone for useful suggestions and stimulating discussions.
This work was supported in part by an EPSRC Doctoral Prize fellowship and by the European Research Council (ERC) through StG ``N-BNP'' 306406.
\setlength{\bibsep}{5pt plus 0.3ex}
\bibliography{bibliography}
%\bibliographystyle{imsart-nameyear}

%\printbibliography

\appendix
\newpage
\section{Additional calculations and proofs}\label{appendix1}
%\begin{proof}[Integrability of $\pi(z)\Zgs(z)$.]\label{sec:integrability_smoothing}
%The integral $\int_{\sX}\pi(z)\Zgs(z)dz$ is finite because $0\leq g(\frac{\pi(y)}{\pi(x)})\leq |g(\frac{\pi(y)}{\pi(z)})-g(1)|+g(1)\leq c|\frac{\pi(y)}{\pi(z)}-1|+g(1)$ where $c$ is the Lipschitz constant of $g$ and therefore
%\begin{multline}
%\int_{\sX}\pi(z)\Zgs(z)dz-g(1)
%\leq
%\int_{\sX^2}\pi(z)\big|g(\frac{\pi(y)}{\pi(z)})-g(1)\big|\Ks(z,dy)dz
%\leq\\
%c\int_{\sX^2}|\pi(y)-\pi(z)|\Ks(z,dy)dz
%\leq\\
%c\int_{\sX^2}\pi(y)\Ks(z,dy)dz
%+
%c\int_{\sX^2}\pi(z)\Ks(z,dy)dz
%\,=\,
%c\int_{\sX^2}\pi(y)\Ks(y,dz)dy+c
%\,=\,2c
%\end{multline}
%\end{proof}
\subsection{Proof of Theorem \ref{thm:loc_bal_iff}}\label{appendix:loc_bal_iff}
To prove Theorem \ref{thm:loc_bal_iff} we first need a technical Lemma to show that a linear bound on $g$ is sufficient to ensure that the integral $\int_{\sX}\pi(x)\Zgs(x)dx$ is well behaved.
\begin{lemma}\label{lemma:convergence_integral}
Let $g:[0,\infty)\rightarrow [0,\infty)$ with $g(1)=1$ and $g(t)\leq a+bt$ for some $a,b\geq 0$ and all $t\geq 0$.
Given a bounded density function $\pi:\sX\rightarrow (0,\infty)$, a family of symmetric kernels $K_\sigma(x,dy)$ such that $K_\sigma(x,dy)\Rightarrow \delta_x(dy)$ for all $x\in\sX$ as $\sigma\downarrow 0$ and $\Zgs(x)$ defined as in \eqref{eq:Z}, it holds
$$
\lim_{\sigma\to 0}
\int_{\sX}\pi(x)\Zgs(x)dx=1\,.%\rightarrow 1\qquad \hbox{ as }\sigma\downarrow 0\,.
$$
\end{lemma}
\begin{proof}[Proof of Lemma \ref{lemma:convergence_integral}]
For every $x\in\sX$ it holds $\Zgs(x)\pi(x)\rightarrow \pi(x)$ as $\sigma\downarrow 0$, see e.g. \eqref{eq:pointwise_Z}.
%\footnote{Not ideal to refer to an equation afterwards...}
Therefore, Fatou's Lemma guarantees 
\begin{equation}\label{eq:liminf}
\liminf_{\sigma\rightarrow 0}\int_{\sX}\pi(x)\Zgs(x)dx
\geq
\int_{\sX}\pi(x)dx
=
1\,.
\end{equation}

We now show that the corresponding $\limsup$ is upper bounded by 1.
To do so, fix $\epsilon>0$ and let $A$ be an open subset of $\sX$ such that $\int_Adx<\infty$ and $\Pi(A)>1-\epsilon$.
%\footnote{Are we sure that such an $A$ exists?}
Using the Bounded Convergence Theorem it can be easily seen that 
%Consider first the integral over $A$, which satisfies 
\begin{equation}\label{eq:integral_A}
\lim_{\sigma\to 0}\int_{A}\pi(x)\Zgs(x)dx= \int_{A}\pi(x)dx=\Pi(A)\,.
\end{equation}
%by the Bounded Convergence Theorem.
In fact $\int_Adx<\infty$ by assumption and $\pi(x)\Zgs(x)$ is bounded over $x$ because 
\begin{multline*}%\label{eq:piZ_bounded}
%\pi(x)\Zgs(x)
%=
\pi(x)\int_{\sX}
g\left(\frac{\pi(y)}{\pi(x)}\right)K_{\sigma}(x,dy)
\leq
\pi(x)\int_{\sX}
\left(a+b\frac{\pi(y)}{\pi(x)}\right)K_{\sigma}(x,dy)=
\\
a\pi(x)+
b\int_{\sX}
\pi(y)K_{\sigma}(x,dy)
\leq (a+b)\sup_{z\in\sX}\pi(z)\,.
\end{multline*}
%Then, fix $\epsilon>0$ and let $A$ be an open subset of $\sX$ such that $\int_Adx<\infty$ and $\Pi(A)>1-\epsilon$.
%By \eqref{eq:piZ_bounded} an the Bounded Convergence Theorem we have 
%\begin{equation}\label{eq:integral_A}
%\lim_{\sigma\to 0}\int_{A}\pi(x)\Zgs(x)dx= \int_{A}\pi(x)dx=\Pi(A)\,.
%\end{equation}
Consider then the integral over $A^c$. 
From $g(t)\leq a+b\,t$ we have
\begin{multline}\label{eq:integral_A_c_1}
\int_{A^c}\pi(x)\Zgs(x)dx
\leq
\int_{A^c}\pi(x)\int_{\sX}\left(a+b\frac{\pi(y)}{\pi(x)}\right)K_\sigma(x,dy)dx
=\\
a\Pi(A^c)
+b\int_{x\in A^c,\,y\in\sX}\pi(y)K_\sigma(x,dy)dx\,.
\end{multline}
Using the reversibility of $K_\sigma$ w.r.t.\ $dx$ we can bound the latter integral as follows
\begin{multline}\label{eq:integral_A_c_2}
\int_{x\in A^c,\,y\in\sX}\pi(y)K_\sigma(x,dy)dx
=
\int_{y\in\sX}\pi(y)\int_{x\in A^c}K_\sigma(y,dx)dy
=\\
\int_{y\in\sX}\pi(y)K_\sigma(y,A^c)dy
\leq
\Pi(A^c)%\int_{y\in A^c}\pi(y)K_\sigma(y,A^c)dy
+\int_{y\in A}\pi(y)K_\sigma(y,A^c)dy\,.
\end{multline}
The term $K_\sigma(y,A^c)$ is upper bounded by $1$ and converges to $0$ for every $y\in A$ because $A^c$ is a closed set and $K_\sigma(y,\cdot)\Rightarrow \delta_{y}(\cdot)$ as $\sigma\downarrow 0$.
Therefore by the Bounded Convergence Theorem $\int_{y\in A}\pi(y)K_\sigma(y,A^c)dy\to 0$ as $\sigma\downarrow 0$.
Combining \eqref{eq:integral_A}, \eqref{eq:integral_A_c_1} and \eqref{eq:integral_A_c_2} we obtain
\begin{equation*}
\limsup_{\sigma\rightarrow 0}\int_{\sX}\pi(x)\Zgs(x)dx
\leq
\Pi(A)
+a\Pi(A^c)+b\Pi(A^c)\leq 1+(a+b)\epsilon
\,.
\end{equation*}
From the arbitrariness of $\epsilon$ it follows 
$\limsup_{\sigma\rightarrow 0}\int_{\sX}\pi(x)\Zgs(x)dx=1$ as desired.
\end{proof}

\begin{proof}[Proof of Theorem \ref{thm:loc_bal_iff}]
\emph{Sufficiency.} Suppose that \eqref{eq:balancing_equation} holds and, without loss of generality, $g(1)=1$.
Using the definition of $\Qgs$ and the symmetry of $\Ks$ it holds
\begin{equation}\label{eq:symm_used}
\frac{\Zgs(x)\Qgs(x,dy)dx}{g\left(\frac{\pi(y)}{\pi(x)}\right)}
\,=\,
\Ks(x,dy)dx
\,=\,
\Ks(y,dx)dy
\,=\,
\frac{\Zgs(y)\Qgs(y,dx)dy}{g\left(\frac{\pi(x)}{\pi(y)}\right)}\,.
\end{equation}
From \eqref{eq:balancing_equation} it follows $g\left(\frac{\pi(x)}{\pi(y)}\right)=\frac{\pi(x)}{\pi(y)}g\left(\frac{\pi(y)}{\pi(x)}\right)$. The latter, together with \eqref{eq:symm_used} implies
\begin{equation*}
\pi(x)\Zgs(x)\Qgs(x,dy)dx
\,=\,
\pi(y)\Zgs(y)\Qgs(y,dx)dy\,.
\end{equation*}
Therefore $\Qgs$ is reversible w.r.t.\ $\frac{\pi(x)\Zgs(x)dx}{\int_{\sX}\pi(z)\Zgs(z)dz}$, where $\int_{\sX}\pi(z)\Zgs(z)dz$ is finite by Lemma \ref{lemma:convergence_integral}.
%Note that %the finiteness of $\int_{\sX}\pi(z)\Zgs(z)dz$ is guaranteed by the Lipschitz-continuity of $g$ (see Appendix \ref{sec:integrability_smoothing} for details).
% because $0\leq g(\frac{\pi(y)}{\pi(x)})\leq |g(\frac{\pi(y)}{\pi(z)})-g(1)|+g(1)\leq c|\frac{\pi(y)}{\pi(z)}-1|+g(1)$ where $c$ is the Lipschitz constant of $g$ and therefore
%\begin{multline*}
%\int_{\sX}\pi(z)\Zgs(z)dz-g(1)
%\leq
%\int_{\sX^2}\pi(z)\big|g(\frac{\pi(y)}{\pi(z)})-g(1)\big|\Ks(z,dy)dz
%\leq\\
%c\int_{\sX^2}|\pi(y)-\pi(z)|\Ks(z,dy)dz
%\leq\\
%c\int_{\sX^2}\pi(y)\Ks(z,dy)dz
%+
%c\int_{\sX^2}\pi(z)\Ks(z,dy)dz
%\,=\,
%c\int_{\sX^2}\pi(y)\Ks(y,dz)dy+c
%\,=\,2c
%\end{multline*}
%\begin{equation*}
%\pi(z)\Zgs(z)
%\quad\leq\quad
%g(1)
%+
%c\int_{\sX}|\pi(y)-\pi(z)|\Ks(z,dy)
%\quad\leq\quad
%g(1)
%+
%c\,\pi(z)
%+
%c(\pi*\Ks)(z)
%\end{equation*}
Since $\Ks(x,\cdot)\Rightarrow\delta_x(\cdot)$ by assumption (see Section \ref{sec:bal_prop_definition}) and since $z\rightarrow g\left(\frac{\pi(z)}{\pi(x)}\right)$ is a bounded and continuous function, it follows that%, as $\sigma\rightarrow 0$
\begin{equation}\label{eq:pointwise_Z}
\Zgs(x)=\int_{\sX}g\left(\frac{\pi(z)}{\pi(x)}\right)K_\sigma(x,dz)
\quad\stackrel{\sigma\downarrow 0}\longrightarrow\quad
\int_{\sX}g\left(\frac{\pi(z)}{\pi(x)}\right)\delta_x(dz)=g(1)
\,,%\qquad\forall x\in\sX\,,
\end{equation}
for every $ x\in\sX$.
Combining \eqref{eq:pointwise_Z} and Lemma \ref{lemma:convergence_integral} we deduce that the probability density function $\frac{\pi(x)\Zgs(x)}{\int_{\sX}\pi(z)\Zgs(z)dz}$ converges pointwise to $\pi(x)$ for every $x\in\sX$.
It follows by Scheff\'{e}'s Lemma that $\frac{\pi(x)\Zgs(x)dx}{\int_{\sX}\pi(z)\Zgs(z)dz}$ converges weakly to $\pi(x)dx$.

\emph{Necessity.}
Suppose \eqref{eq:balancing_equation} does not hold and therefore $g(\tzero)\neq \tzero g(\frac{1}{\tzero})$ for some $\tzero>0$.
To prove that \eqref{eq:balancing_equation} is a necessary condition for $\{\Qgs\}_{\sigma>0}$ to be locally-balanced with respect to a general $\Pi(dx)=\pi(x)dx$ with bounded and continuous density $\pi$, it is enough to provide a counterexample.
Consider then $\sX=\{0,1\}$, $dx$ being the counting measure, %$\Ks$ any kernel satisfying Assumption \ref{ass:base_kernel},
 $\pi(0)=\frac{1}{1+\tzero}$ and $\pi(0)=\frac{\tzero}{1+\tzero}$.
%Here $dx$ is any reference measure on $\sX$ with positive mass on both states, $0$ and $1$,
From \eqref{eq:symm_used} it follows that $\Qgs$ is reversible with respect to $\Pi_\sigma(dx)=\pi_\sigma(x)dx$, where $\pi_\sigma=(\pi_\sigma(1),\pi_\sigma(2))$ is proportional to 
$(\frac{\pi(0)\Zgs(0)}{g(\tzero)},\frac{\pi(1)\Zgs(1)}{\tzero\,g(\frac{1}{\tzero})})$.
From $\Ks(x,\cdot)\Rightarrow\delta_x(\cdot)$ it follows that $\Pi_\sigma(dx)\Rightarrow\Pi_0(dx)=\pi_0(x)dx$ with $\pi_0=(\pi_0(1),\pi_0(2))$ proportional to 
$(\frac{\pi(0)}{g(\tzero)},\frac{\pi(1)}{\tzero\,g(\frac{1}{\tzero})})$.
Finally, $g(\tzero)\neq \tzero g(\frac{1}{\tzero})$ implies that $\Pi_0\neq\Pi$ and thus $\Pi_\sigma\nRightarrow\Pi$.
\end{proof}

\subsection{Proof of Theorem \ref{thm:peskun_constant}}
To prove part (a) of Theorem \ref{thm:peskun_constant} we need the following Lemma (see also \citet[Corollary 1]{Latuszynski2013} for a similar result in the context of general state spaces $\sX$).%proof of Lemma \ref{lemma:Peskun_ext} in the context of time-sampled Markov chain or Appendix \ref{appendix:proof_lemma} for a more direct proof.
\begin{lemma}\label{lemma:Peskun_ext}
Let $P$ be a $\Pi$-reversible Markov transition kernels on a finite space $\sX$.
Let $\tilde{P}=c\,P+(1-c)\Id$, where $\Id$ is the identity kernel and $c\in(0,1]$.
Then it holds
\begin{equation*}
\var_{\pi}(h,\tilde{P})
=
\frac{\var_{\pi}(h,P)}{c}+\frac{1-c}{c}\,\var_\pi(h)
\qquad \forall h\in L^2(\Pi)\,.
\end{equation*}
\end{lemma}
\begin{proof}[Proof of Lemma \ref{lemma:Peskun_ext}]\label{appendix:proof_lemma}
%\footnote{Maybe we could actually remove this proof altogether and refer to LatuszynskiRoberts2013 directly...}
Suppose $\E_{\pi}[h]=0$ (otherwise consider $h-\E_\pi[h]$).
Also, if $P$ is reducible then the statement is trivially true, so suppose $P$ to be irreducible.
Let $\{(\lambda_i,f_i)\}_{i=1}^n$ and $\{(\tilde{\lambda}_i,\tilde{f}_i)\}_{i=1}^n$ be eigenvalues and eigenfunctions of $P$ and $\tilde{P}$ respectively.
We can take $\{f_i\}_{i=1}^n$ and $\{\tilde{f}_i\}_{i=1}^n$ to be an orthonormal basis of $L^2(\R^\sX,\pi)$ with $\lambda_1=\tilde{\lambda}_1=1$, $f_1=\tilde{f}_1=(1,\dots,1)^T$ and $-1\leq \lambda_i,\tilde{\lambda}_i< 1$ for $i\geq 2$ \cite[Lemmas 12.1,12.2]{Levin2009}.
Also, from the definition of $\tilde{P}$, it follows that we can take $\tilde{f}_i=f_i$ and $\tilde{\lambda}_i=c\,\lambda_i+(1-c)$.
The asymptotic variances can be written as
\begin{equation}\label{eq:av_representation}
\var_{\pi}(h,P)=\sum_{i=2}^n\frac{1+\lambda_i}{1-\lambda_i}\E_{\pi}[h\,f_i]^2
\quad\hbox{ and }\quad
\var_{\pi}(h,\tilde{P})=\sum_{i=2}^n\frac{1+\tilde{\lambda}_i}{1-\tilde{\lambda}_i}\E_{\pi}[h\,\tilde{f}_i]^2\,.
\end{equation}
%and %$\var_{\pi}(h,\tilde{P})=\sum_{i=2}^n\frac{1+\tilde{\lambda}_i}{1-\tilde{\lambda}_i}\E_{\pi}[h\,\tilde{f}_i]^2$
For \eqref{eq:av_representation} see for example the proofs of \citet[Theorem 1]{Mira2001} and \citet[Lemmas 12.20]{Levin2009}.
Rearranging $\tilde{\lambda}_i=c\,\lambda_i+(1-c)$ we obtain
$\frac{1+\tilde{\lambda}_i}{1-\tilde{\lambda}_i}=\frac{1}{c}\frac{1+\lambda_i}{1-\lambda_i}+\frac{1-c}{c}$ for $i\geq 2$.
Thus 
\begin{equation}\label{eq:peskun_ext_1}
\var_{\pi}(h,\tilde{P})
=
\sum_{i=2}^n\frac{1+\tilde{\lambda}_i}{1-\tilde{\lambda}_i}\E_{\pi}[h\,f_i]^2
=
\frac{1}{c}\sum_{i=2}^n\frac{1+\lambda_i}{1-\lambda_i}\E_{\pi}[h\,f_i]^2
+\frac{1-c}{c}\sum_{i=2}^n\E_{\pi}[h\,f_i]^2\,.
\end{equation}
Since $\{f_i\}_{i=1}^n$ form an orthonormal basis of $L^2(\R^\sX,\pi)$ and $\E_{\pi}[h\,f_1]=\E_{\pi}[h]=0$, then $\sum_{i=2}^n\E_{\pi}[h\,f_i]^2=\E_{\pi}[h^2]=\var_\pi(h)$.
Therefore \eqref{eq:peskun_ext_1} becomes
$\var_{\pi}(h,\tilde{P}_1)
=
\frac{1}{c}\cdot\var_{\pi}(h,P_1)+\frac{1-c}{c}\var_\pi(h)\,.$
\end{proof}

\begin{proof}[Proof of Theorem \ref{thm:peskun_constant}]
Part (a), case $c>1$: define $\tilde{P}_1=\frac{1}{c}P_1+(1-\frac{1}{c})\I_n$.
From Lemma \ref{lemma:Peskun_ext} it follows
$\var_{\pi}(h,\tilde{P}_1)
=
c\cdot\var_{\pi}(h,P_1)+(c-1)\,\var_\pi(h)$
or, equivalently,
$\var_{\pi}(h,P_1)
=
\frac{1}{c}\var_{\pi}(h,\tilde{P}_1)+\frac{1-c}{c}\,\var_\pi(h)$.
Since $\tilde{P}_1(x,y)\geq P_2(x,y)$ for $x\neq y$, by Theorem \citet[Thm.2.1.1]{Peskun1973} it holds $\var_\pi(h,\tilde{P}_1)\leq\var_\pi(h,P_2)$.
Therefore
$$
\var_{\pi}(h,P_1)
\;=\;
\frac{\var_{\pi}(h,\tilde{P}_1)}{c}+\frac{1-c}{c}\,\var_\pi(h)
\;\leq\;
\frac{\var_{\pi}(h,P_2)}{c}+\frac{1-c}{c}\,\var_\pi(h)\,.
$$
Part (a), case $c\leq1$:
define $\tilde{P}_2=c\,P_2+(1-c)\I_n$.
From \citet[Thm.2.1.1]{Peskun1973} and $P_1(x,y)\geq \tilde{P}_2(x,y)$ for $x\neq y$ it follows
$\var_\pi(h,P_1)\leq\var_\pi(h,\tilde{P}_2)$.
From Lemma \ref{lemma:Peskun_ext} it follows
$\var_{\pi}(h,\tilde{P}_2)
=
\frac{1}{c}\var_{\pi}(h,P_2)+\frac{1-c}{c}\var_\pi(h)\,.$
The latter equality and $\var_\pi(h,P_1)\leq\var_\pi(h,\tilde{P}_2)$ provide us with part (a).\\
Part (b) follows from the definition of $Gap(P)$.
\end{proof}

\subsection{Proof of Theorem \ref{thm:peskun_result}}
In this section we state and prove a slightly more general result than Theorem \ref{thm:peskun_result}, namely Theorem \ref{thm:peskun_lower_bound}.
The latter applies also to continuous spaces and include both upper and lower bounds.
Theorem \ref{thm:peskun_result} is kept in the main body of the paper for simplicity of exposition.

As in the beginning of Section \ref{sec:bal_prop_definition}, let $\Pi(dx)=\pi(x)dx$ be a target probability distribution on some topological space $\sX$ with bounded density $\pi$ with respect to some reference measure $dx$ and let $K(x,dy)$ be a symmetric Markov kernel on $\sX$.
Define $\RR\subseteq\sX\times\sX$ as the set where the two measures $\Pi(dx)K(x,dy)$ and $\Pi(dy)K(y,dx)$ are mutually absolutely continuous (\citet[Prop.1]{Tierney1998} shows that $\RR$ is unique up to sets of zero mass for both measures) and for any $(x,y)$ in $\RR$ denote by $t_{xy}$ the Radon-Nikodym derivative $\frac{\Pi(dy)K(y,dx)}{\Pi(dx)K(x,dy)}$.
Note that $\RR$ is also the set where $\Pi(dx)Q_g(x,dy)$ and $\Pi(dy)Q_g(y,dx)$ are mutually absolutely continuous.
Define
\begin{equation}\label{eq:smoothness0}
\bal{g}
\quad=\quad
\sup
_{(x,y)\in\RR}%_{x\in\sX,\;y: k(x,y)>0}
\frac{g(t_{xy})}{t_{xy}g(t_{yx})}
\;\geq\;1
\;,
\end{equation}
%where for any $(x,y)\in\RR$ we define $t_{xy}=\frac{\Pi(dy)Q(y,dx)}{\Pi(dx)Q(x,dy)}$, 
and
\begin{equation}\label{eq:smoothness_cont}
\smooth{g} 
\quad=\quad
\sup%\left\{
_{(x,y)\in\RR}%_{x\in\sX,\;y: k(x,y)>0}
\frac{Z_g(y)}{Z_g(x)}
%\;|\;x\in\sX,\;y\in N(x)\right\}
\;\geq\;1
\;.
\end{equation}
The suprema in \eqref{eq:smoothness0} and \eqref{eq:smoothness_cont} have to be intended $\Pi(dx)K(x,dy)$-almost everywhere.
Both $\bal{g}$ and $\smooth{g} $ are greater or equal than one because, by inverting $x$ and $y$ both the fraction in \eqref{eq:smoothness0} and \eqref{eq:smoothness_cont} get inverted.
The constant $\bal{g}$ represents how ``unbalanced'' the function $g$ is: the bigger $\bal{g}$ the less balanced $g$ is according to \eqref{eq:balancing_equation} (if $\bal{g}=1$ then $g$ satisfies \eqref{eq:balancing_equation}).
It is easy to see that the following result implies Theorem \ref{thm:peskun_result} as a special case.
\begin{theorem}\label{thm:peskun_lower_bound}
Let $g:(0,\infty)\rightarrow(0,\infty)$.
Define $\tilde{g}(t)=\min\{g(t),t\,g(1/t)\}$ and let $P_g$ and $P_{\tilde{g}}$ be the MH kernels obtained from the proposals $Q_g$ and $Q_{\tilde{g}}$ respectively (see \eqref{eq:informed_proposal} for definition).
For $\Pi$-almost every $x$ it holds%It holds
\begin{equation}\label{eq:peskun_lower_bound}
\frac{P_{\tilde{g}}(x,A)}{\smooth{g} \smooth{\tilde{g}}\bal{g}}
\quad\leq\quad 
P_{g}(x,A)
\quad\leq\quad 
(\smooth{g} \smooth{\tilde{g}})P_{\tilde{g}}(x,A)
\qquad
\forall A\not\owns x\,.%\in \sX\,,
\end{equation}
\end{theorem}
\begin{proof}[Proof of theorem \ref{thm:peskun_lower_bound}]
Let $x\notin A\subseteq \sX$.
By construction, $P_g(x,\cdot)$ is absolutely continuous with respect to $K(x,\cdot)$ on $\sX\backslash\{x\}$.
In particular, given $y\in\RR$ with $y\neq x$% and denoting for brevity $\txy=\frac{\pi(y)}{\pi(x)}$ and $\tyx=\frac{\pi(x)}{\pi(y)}$
, the Radon-Nikodym derivative between $P_g(x,\cdot)$ and $K(x,\cdot)$ at $y$ satisfies
\begin{equation}\label{eq:P_g_over_K}
\frac{P_g(x,dy)}{K(x,dy)}\;=\;
\frac{g(\txy)}{Z_g(x)}\,
\min\left\{1\,,\,
\txy\,
\frac{g(\tyx)}{Z_g(y)}\frac{Z_g(x)}{g(\txy)}
\right\}
\;=\;
\min\left\{
\frac{g(\txy)}{Z_g(x)}\,,\,
\frac{\txy g(\tyx)}{Z_g(y)}
\right\}\,.
\end{equation}
Using \eqref{eq:P_g_over_K} and the definition of $\smooth{g}$, it follows that for $\Pi$-almost every $x$
\begin{equation}\label{eq:control_density}
\frac{1}{\smooth{g}}\frac{\min\left\{
g(\txy),\txy g(\tyx)\right\}
}{Z_g(x)}
\quad\leq\quad
\frac{P_g(x,dy)}{K(x,dy)}
\quad\leq\quad
\smooth{g}\frac{\min\left\{
g(\txy),\txy g(\tyx)\right\}
}{Z_g(x)}\,.
\end{equation}

%By definition of $\tilde{g}$, it holds $g(t)\geq\tilde{g}(t)$ for any $t>0$ and thus $Z_{g}(x)\geq Z_{\tilde{g}}(x)$ for any $x$ in $\sX$.
Using \eqref{eq:control_density}, the definition of $\tilde{g}$ and the inequality $Z_{g}(x)\geq Z_{\tilde{g}}(x)$ (which follows from $g(t)\geq\tilde{g}(t)$) we can deduce that for $\Pi$-almost every $x$
\begin{multline}\label{eq:peskun_proof_1}
P_g(x,A)
\;\leq\;
%\smooth{g}
%\int_{A}\frac{\min\{
%g(\txy),
%\txy\,g(\tyx)
%\}
%}{Z_g(x)}
%K(x,dy)
%\;=\;
\smooth{g} \int_{A}
\frac{\tilde{g}(\txy)}{Z_g(x)}K(x,dy)
\;\leq\;
\smooth{g} \int_{A}
\frac{\tilde{g}(\txy)}{Z_{\tilde{g}}(x)}K(x,dy)\,,
\end{multline}
Using the analogous of \eqref{eq:control_density} for $\tilde{g}$ rather than $g$, we have
\begin{multline}\label{eq:peskun_proof_2}
P_{\tilde{g}}(x,A)
\;\geq\;
\frac{1}{\smooth{\tilde{g}}}\int_{A}
\frac{\min\left\{
{\tilde{g}}(\txy)\,,\,\txy {\tilde{g}}(\tyx)
\right\}}{Z_{\tilde{g}}(x)}
K(x,dy)
\;=\;
\frac{1}{\smooth{\tilde{g}}} \int_{A}
\frac{\tilde{g}(\txy)}{Z_{\tilde{g}}(x)}K(x,dy)\,.
\end{multline}
The upper bound in \eqref{eq:peskun_lower_bound} follows from \eqref{eq:peskun_proof_1} and \eqref{eq:peskun_proof_2}.
To obtain the lower bound in \eqref{eq:peskun_lower_bound} first note that using the lower bound in \eqref{eq:control_density} and the definition of $\bal{g}$ for $\Pi$-almost every $x$ it holds
\begin{multline}\label{eq:peskun_proof_3}
P_g(x,A)
\;\geq\;
\frac{1}{\smooth{g}} \int_{A}
\frac{\min\left\{
g(\txy),\txy g(\tyx)\right\}}{Z_g(x)}K(x,dy)
\;\geq\;
\frac{1}{\bal{g}\smooth{g}} \int_{A}
\frac{g(\txy)}{Z_g(x)}K(x,dy)
\,.
\end{multline}
Then note that 
\begin{multline}
P_{\tilde{g}}(x,A)\;=\;
\int_{A}\min\left\{
\frac{{\tilde{g}}(\txy)}{Z_{\tilde{g}}(x)}\,,\,
\frac{\txy {\tilde{g}}(\tyx)}{Z_{\tilde{g}}(y)}
\right\}
K(x,dy)
\;\leq\;
\frac{1}{\smooth{\tilde{g}}} \int_{A}
\frac{\tilde{g}(\txy)}{Z_{\tilde{g}}(x)}K(x,dy)\,.
\end{multline}
\end{proof}

\subsection{Proof of Proposition \ref{thm:norm_const}}
\begin{proof}[Proof of Proposition \ref{thm:norm_const}]
Consider Example \ref{ex:perfect_matchings}. 
Fix $\rho\in\Sn$ and $\rho'=\rho\circ(i_0,j_0)$ for some $i_0,j_0\in\{1,\dots,n\}$, with $i_0<j_0$.
Denoting $g\left(\frac{\pi^{(n)}(\rho\circ(i,j))}{\pi^{(n)}(\rho)}\right)=g\left(\frac{w_{i\rho(j)}w_{j\rho(i)}}{w_{i\rho(i)}w_{j\rho(j)}}\right)$ by $g^\rho_{ij}$ it holds
$$%\begin{multline}
Z^{(n)}_g(\rho)
\;=\;
\sum_{i,j=1,\,i<j}^n g^\rho_{ij}
\;=\;
\sum_{\substack{i,j=1,\,i<j\\ \{i,j\}\cap\{i_0,j_0\}=\emptyset}}^n g^\rho_{ij}+
\sum_{\substack{i,j=1,\,i<j\\ \{i,j\}\cap\{i_0,j_0\}\neq\emptyset}}^n g^\rho_{ij}\,.
$$
Given 
$I=\left[\frac{\inf_{i,j}w_{ij}^2}{\sup_{i,j}w_{ij}^2},\frac{\sup_{i,j}w_{ij}^2}{\inf_{i,j}w_{ij}^2}\right]$, $\underline{g}=\inf_{t\in I} g(t)$ and $\overline{g}=\sup_{t\in I} g(t)$ it holds
$
\underline{g}\leq
g^\rho_{ij}
\leq
\overline{g}$.
Note that $\underline{g}>0$ and $\overline{g}<\infty$ because $g$ and $1/g$ are locally bounded and $I$ is compact.
Therefore
$$
\sum_{\substack{i,j=1,\,i<j\\ \{i,j\}\cap\{i_0,j_0\}=\emptyset}}^n g^\rho_{ij}
\;\geq\;
\sum_{\substack{i,j=1,\,i<j\\ \{i,j\}\cap\{i_0,j_0\}=\emptyset}}^n \underline{g}
\;=\;
\left(\frac{n(n-1)}{2}-(2n-3)\right)\underline{g}
=O(n^2)
$$
and
$$
\sum_{\substack{i,j=1,\,i<j\\ \{i,j\}\cap\{i_0,j_0\}\neq\emptyset}}^n g^\rho_{ij}
\;\leq\;
\sum_{\substack{i,j=1,\,i<j\\ \{i,j\}\cap\{i_0,j_0\}\neq\emptyset}}^n \overline{g}
\;=\;
(2n-3)\overline{g}
=O(n)\,.
$$
If follows that %$Z^{(n)}_g(\rho)$ is asymptotic to $\sum_{\{i,j\}\cap\{i_0,j_0\}=\emptyset} g^\rho_{ij}$. Thus
$$
\lim_{n\rightarrow\infty}\frac{Z^{(n)}_g(\rho')}{Z^{(n)}_g(\rho)}
\;=\;
\lim_{n\rightarrow\infty}\frac{\sum_{\{i,j\}\cap\{i_0,j_0\}=\emptyset} g^\rho_{ij}}{\sum_{\{i,j\}\cap\{i_0,j_0\}=\emptyset} g^{\rho'}_{ij}}
\;=\;
\lim_{n\rightarrow\infty}\frac{\sum_{\{i,j\}\cap\{i_0,j_0\}=\emptyset} g^\rho_{ij}}{\sum_{\{i,j\}\cap\{i_0,j_0\}=\emptyset} g^\rho_{ij}}
\;=\;1\,,
$$
where $g^\rho_{ij}=g^{\rho'}_{ij}$ for $\{i,j\}\cap\{i_0,j_0\}=\emptyset$ because $\rho'(i)=\rho(i)$ for $i\in\{1,\dots,n\}\backslash\{i_0,j_0\}$.

The proofs for Example \ref{ex:binary} and \ref{ex:ising} are analogous.
\end{proof}

%%%%%%%%%%%
%%%%%%%%%%%
%\subsection{Proof of Theorem \ref{theorem:limit_hypercube}}\label{appendix:hypercube}
%\subsection{Supplement to Section \ref{sec:choice_of_g}}
\subsection{Proof of Theorem \ref{theorem:limit_hypercube}}\label{appendix:hypercube}
To prove Theorem \ref{theorem:limit_hypercube} we first need the following Lemma.
\begin{lemma}\label{lemma:core}
For every positive integers $k<n$, let
\begin{multline}\label{eq:def_R_n}
R^{(k)}_{n}=\Big\{
(x_{k+1},\dots,x_{n})\in\{0,1\}^{n-k}\,:\\
\left|\frac{Z^{(n)}_g(\bx_{1:n})}{n}-\E\left[\frac{Z^{(n)}_g(\textbf{X}_{1:n})}{n}\right]\right|\leq \frac{1}{n^{1/4}}
\;\;\forall\,(x_{1},\dots,x_k)\in\Omega^{(k)}
\Big\}\,.
\end{multline}
Then it holds $\lim_{n\rightarrow\infty}\pi^{(n)}(\{0,1\}^k\times R^{(k)}_n)= 1$.
\end{lemma}
\begin{proof}[Proof of Lemma \ref{lemma:core}]
From \eqref{eq:hypercube_proposal_general} it follows
%Using the formulation of $Q^{(n)}_g(\bx_{1:n},\cdot)$ in \eqref{eq:hypercube_proposal_general}, the normalizing constant $Z^{(n)}_g(\bx_{1:n})$ can be written as
\begin{equation*}
Z^{(n)}_g(\bx_{1:n})
\;=\;
\sum_{i=1}^nv_i\,\big(c_i(1-p_i)(1-x_i)+(1-c_i)p_ix_i\big)\,.
\end{equation*}
Thus, for $\X_{1:n}\sim\pi^{(n)}$, it holds $\E\left[Z^{(n)}_g(\X_{1:n})\right]=\sum_{i=1}^nv_i\,p_i(1-p_i)$.
Using the triangular inequality we can split 
\begin{equation}\label{eq:split}
\left|\frac{Z^{(n)}_g(\bx_{1:n})}{n}-\E\left[\frac{Z^{(n)}_g(\textbf{X}_{1:n})}{n}\right]\right|
\leq
\frac{\left|\sum_{i=1}^{k} m_i\right|}{n}
+
\frac{\left|\sum_{i=k+1}^{n} m_i\right|}{n}\,,
\end{equation}
where $m_i=v_i\,\big(c_i(1-p_i)(1-x_i)+(1-c_i)p_ix_i-p_i(1-p_i)\big)$.
Note that each $m_i$ is upper bounded by $2 \sup_{i\in\N}v_i$ because $p_i$, $x_i$ and $c_i$ belong to $[0,1]$, while the finiteness of $\sup_{i\in\N}v_i$ can be deduced by the definition of $v_i$'s and the assumptions $\sup_{i\in \N}p_i<1$ and $\inf_{i\in \N}p_i>0$.
Therefore $\frac{\left|\sum_{i=1}^{k} m_i\right|}{n}\leq \frac{2 k \sup_{i\in\N}v_i}{n}$ which will eventually be smaller than $\frac{1}{2n^{1/4}}$ as $n$ increases.
Consider now last sum in \eqref{eq:split}, $\frac{\left|\sum_{i=k+1}^{n} m_i\right|}{n}$, which depends on $\bx_{1:n}$ only through the $x_i$'s with $i>k$.
For $\X_{1:n}\sim\pi^{(n)}$ and $M_i=v_i\,\big(c_i(1-p_i)(1-X_i)+(1-c_i)p_iX_i-p_i(1-p_i)\big)$, the random variable $\frac{\left|\sum_{i=k+1}^{n} M_i\right|}{n}$ has mean zero and variance upper bounded by $\frac{\sup_{i\in \N} v_i^2}{n-k}$.
Therefore, using for example the Markov inequality, we have $P(\frac{\left|\sum_{i=k+1}^{n} M_i\right|}{n}\geq \frac{1}{2n^{1/4}})\to 0$ as $n$ increases.
It follows the thesis.
\end{proof}

\begin{proof}[Proof of Theorem \ref{theorem:limit_hypercube}]
Let $k$ be fixed and denote by $A^{(n)}(\bx_{1:k},\by_{1:k})$ the jumping rates of $S^{(n)}_{1:k}$ from $\bx_{1:k}$ to $\by_{1:k}$.
By construction, for any $\by_{1:k}\neq \bx_{1:k}$ and $\by_{1:k}\notin N\big(\bx_{1:k}\big)$, it holds $A^{(n)}(\bx_{1:k},\by_{1:k})=0$.
Also, it is easy to see that, if $x_i=0$, it holds
\begin{equation*}
A^{(n)}(\bx_{1:k},\bx_{1:k}+\e^{(i)}_{1:k})
\;=\;
\frac{v_i(1-p_i)}{\frac{Z^{(n)}_g(\bx_{1:n})}{n}}
\left(c_i\wedge\left((1-c_i)
\frac{Z^{(n)}_g(\bx_{1:n})}{Z^{(n)}_g(\bx_{1:n}+\e^{(i)}_{1:n})}
\right)\right)\,,
\end{equation*}
%\begin{multline*}
%A^{(n)}(\bx_{1:k},\bx_{1:k}+\e^{(i)}_{1:k})
%\;=\\
%n\;Q^{(n)}(\bx_{1:n},\bx_{1:n}+\e^{(i)}_{1:n})\left(1\wedge\frac{\pi^{(n)}(\bx_{1:n}+\e^{(i)}_{1:n})Q^{(n)}(\bx_{1:n}+\e^{(i)}_{1:n},\bx_{1:n})}{\pi^{(n)}(\bx_{1:n})Q^{(n)}(\bx_{1:n},\bx_{1:n}+\e^{(i)}_{1:n})}\right)
%\;=\\
%\frac{v_i(1-p_i)}{\frac{Z^{(n)}_g(\bx_{1:n})}{n}}
%\left(c_i\wedge\left((1-c_i)
%\frac{Z^{(n)}_g(\bx_{1:n})}{Z^{(n)}_g(\bx_{1:n}+\e^{(i)}_{1:n})}
%\right)\right)\,,
%\end{multline*}
while if $x_i=1$
\begin{equation*}
A^{(n)}(\bx_{1:k},\bx_{1:k}-\e^{(i)}_{1:k})
\;=\;
\frac{v_ip_i}{\frac{Z^{(n)}_g(\bx_{1:n})}{n}}
\left((1-c_i)\wedge\left(c_i
\frac{Z^{(n)}_g(\bx_{1:n})}{Z^{(n)}_g(\bx_{1:n}-\e^{(i)}_{1:n})}
\right)\right)\,.
\end{equation*}
Note that $S^{(n)}_{1:k}$ is not a Markov process and indeed the jumping rates % because the jumping rates $A^{(n)}(\bx_{1:k},\bx_{1:k}+\e^{(i)}_{1:k})$ and $A^{(n)}(\bx_{1:k},\bx_{1:k}-\e^{(i)}_{1:k})$ 
depend also on the last $(n-k)$ components $(x_{k+1},\dots,x_{n})$.
We now show that, given $R_n^{(k)}$ as in Lemma \ref{lemma:core}, it holds
\begin{equation}\label{eq:sup_core}
\sup_{\bx_{1:n}\in\{0,1\}^k\times R_n^{(k)}\,,\,\by_{1:k}\in\{0,1\}^k}|A^{(n)}(\bx_{1:k},\by_{1:k})-A(\bx_{1:k},\by_{1:k})|
 \quad\stackrel{n\rightarrow\infty}\longrightarrow\quad 0\,.
\end{equation}
Equation \eqref{eq:sup_core}, together with $\lim_{n\rightarrow\infty}\pi^{(n)}(\{0,1\}^k\times R^{(k)}_n)= 1$ from Lemma \ref{lemma:core}, implies that that $S^{(n)}_{1:k} \stackrel{n\rightarrow\infty}\Longrightarrow S_{1:k}$, using Corollary 8.7 from \citealp[Ch.4]{Ethier1986}.
Suppose first $x_i=0$. In this case it holds
\begin{multline*}
|A^{(n)}(\bx_{1:k},\bx_{1:k}+\e^{(i)}_{1:k})-A(\bx_{1:k},\bx_{1:k}+\e^{(i)}_{1:k})|
\;=\\
v_i(1-p_i)\left|
\frac{1}{\frac{Z^{(n)}_g(\bx_{1:n})}{n}}
\left(c_i\wedge\left((1-c_i)
\frac{Z^{(n)}_g(\bx_{1:n})}{Z^{(n)}_g(\bx_{1:n}+\e^{(i)}_{1:n})}
\right)\right)
-
\frac{c_i\wedge(1-c_i)}{\bar{Z}}
\right|\,.
\end{multline*}
Adding and subtracting $\frac{(c_i\wedge(1-c_i))}{Z^{(n)}_g(\bx_{1:n})/n}$, and using $(1-p_i)\leq 1$, the latter expression is upper bounded by 
\begin{multline}\label{eq:A_bound}
v_i\Bigg(
\frac{1}{\frac{Z^{(n)}_g(\bx_{1:n})}{n}}
\left|
\left(c_i\wedge\left((1-c_i)
\frac{Z^{(n)}_g(\bx_{1:n})}{Z^{(n)}_g(\bx_{1:n}+\e^{(i)}_{1:n})}
\right)\right)
-
\left(c_i\wedge(1-c_i)\right)
\right|
+\\
\left|\frac{(c_i\wedge(1-c_i))}{\frac{Z^{(n)}_g(\bx_{1:n})}{n}}-\frac{(c_i\wedge(1-c_i))}{\bar{Z}}\right|
\Bigg)\,.
\end{multline}
Defining $\alpha_n=\frac{1}{n^{1/4}}+\left|\E\left[\frac{Z^{(n)}_g(\textbf{X}_{1:n})}{n}\right]-\bar{Z}\right|$, so that $\left|\frac{Z^{(n)}_g(\bx_{1:n})}{n}-\bar{Z}\right|
\;\leq\;
\alpha_n$, the expression in \eqref{eq:A_bound} is in turn bounded by
\begin{equation}\label{eq:bound_for_A}
(\sup_{i\in\N}v_i)\left(
\frac{1}{\bar{Z}-\alpha_n}
\left|\frac{2\alpha_n}{\bar{Z}-\alpha_n}\right|
+
\left|\frac{\alpha_n}{\bar{Z}\,(\bar{Z}-\alpha_n|)}\right|
\right)\,.
%\quad\stackrel{n\rightarrow\infty}\longrightarrow\quad 0\,,
\end{equation}
The case $x_i=1$ is analogous.
The expression in \eqref{eq:bound_for_A} does not depend on $\x_{1:n}$ and converges to 0 as $n$ increases because $\lim_{n\to \infty}\alpha_n= 0$. Thus \eqref{eq:sup_core} holds as desired and $S^{(n)}_{1:k} \stackrel{n\rightarrow\infty}\Longrightarrow S_{1:k}$.
\end{proof}

\subsubsection{Some references for mixing time of product chains}\label{appendix:product_chains}
For each $k\in\N$, $S_{1:k}$ is a continuous time Markov chain with independent components. 
Such chains are often called product chains and received considerable attention, for example in the context of the \emph{cutoff} phenomenon (see e.g. \citet[Thm.2.9]{Diaconis1996} or \citet[Ch.20.4]{Levin2009} and references therein). 
In particular \citet[Thm.20.7]{Levin2009}, \citet[Prop. 7]{Barrera2006} and \citet[Cor. 4.3]{Bon2001} provide results concerning the mixing time of $\{S_{1:k}\}_{k=1}^{\infty}$.
Such results tell us that, in the case of a sequence of independent binary Markov processes like $\{S_{1:k}\}_{k=1}^{\infty}$, the asymptotic mixing time depends on the flipping rates of the worst components (provided they are a non-negligible quantity), and in particular in our case the mixing time is minimized by maximizing the quantity $\bar{Z}(\textbf{v})^{-1}\liminf_{i\rightarrow\infty} v_i$ (see Theorem \ref{theorem:limit_hypercube} for the definition of $\bar{Z}(\textbf{v})$ and \citet{Barrera2006} and \citet{Bon2001} for the precise assumptions on the flipping rates).
It can be seen that the latter quantity is maximized by choosing $v_i$ to be constant over $i$, meaning $v_i=\bar{v}$  for any $i\in\N$ for some $\bar{v}>0$.

\newpage
\section{Supplement for the simulation studies}\label{appendix:simulations}
\subsection{Supplement to Section \ref{sec:sampling_matchings}}\label{supp:matchings}
Figure \ref{fig:traceplots_iid_all} provides additional tracepots, acceptance rates and effective sample sizes related to the simulation study in Section \ref{sec:sampling_matchings}.
\begin{figure}[h!]
\includegraphics[width=0.98\linewidth]{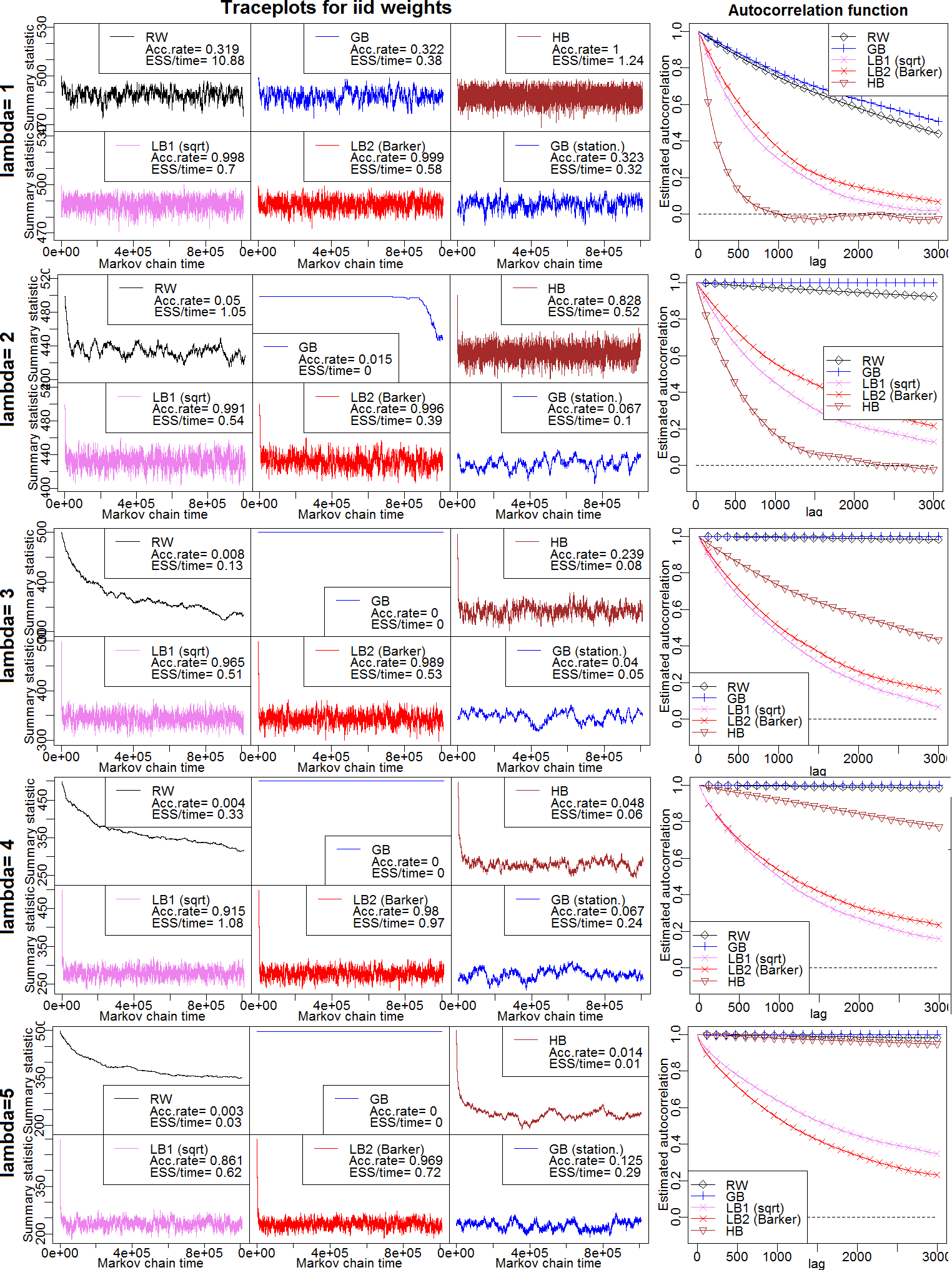}
\caption{Performances for the five MCMC of Section \ref{sec:sampling_matchings} with targets as in Example \ref{ex:perfect_matchings}, $n=500$ and $\log (w_{ij})\stackrel{iid}\sim N(0,\lambda^2)$.
Left: traceplots of a summary statistic (Hamming distance from fixed permutation). 
Right: estimated autocorrelation functions.}\label{fig:traceplots_iid_all}
\end{figure}
Considering effective sample sizes per unit of computation time as measure of efficiency, it can be seen that the uninformed RW proposal performs best for flatter targets (i.e.\ small values of $\lambda$), while as the roughness of the target increases (i.e.\ larger values of $\lambda$) locally-balanced schemes strongly outperforms other schemes under consideration. Note that HB performs slightly better than locally-balanced schemes for flatter targets, but collapses as the target becomes more challenging.

\subsection{Supplement to Section \ref{sec:sampling_ising}}\label{supp:Ising}
In Section \ref{sec:sampling_ising} we tested the MCMC schemes under consideration using Ising model distributions as a benchmark.
In particular we considered distributions motivated by Bayesian image analysis.
In this context the variables $x_i$ represent the classification of pixel $i$ as ``object'' ($x_i=1$) or ``background'' ($x_i=-1$). In a more general multiple-objects context one would use the Potts model, which is the direct generalization of the Ising model to variables $x_i$'s taking values in $\{0,1,\dots,k\}$ for general $k$.
%Given the allocation variables $\{x_i\}_{i\in V_n}$, the observed image $\{y_i\}_{i\in V_n}$ is generated from the mixture model% with one distribution for background pixels and a distinct one for object's pixels, meaning that
%$$
%y_i|\{x_j\}_{j\in V_n},\{y_j\}_{j\in V_n,j\neq i}
%\;\sim\;\left\{
%\begin{array}{ll}
%p_{obj}(y_i)& \hbox{if }x_i=1\,,\\
%p_{back}(y_i) & \hbox{if }x_i=-1\,.
%\end{array}
%\right.
%$$
%where $p_{obj}(\cdot)$ is the distribution of object's pixels and $p_{back}(\cdot)$ the background's one.
%%$$
%%y_i|\{x_j\}_{j\in V_n},\{y_j\}_{j\in V_n,j\neq i}
%%\;\sim\;\left\{
%%\begin{array}{ll}
%%p_{+1}(y_i)& \hbox{if }x_i=1\,,\\
%%p_{-1}(y_i) & \hbox{if }x_i=-1\,.
%%\end{array}
%%\right.
%%$$
%Thus, pixels are conditionally independent given the allocation variables and the global parameters.
%A typical example would be $y_i|\{x_i\}_{i\in V_n},\{y_i\}_{i\in V_n}$
Given the allocation variables $\{x_i\}_{i\in V_n}$, the observed image's pixels $\{y_i\}_{i\in V_n}$ are conditionaly independent and follow a mixture model
$p(y_i|\{x_j\}_{j\in V_n},\{y_j\}_{j\in V_n,j\neq i})=p(y_i|x_i)$ where $p(y_i|x_i=-1)$ is a distribution specific to background pixels and $p(y_i|x_i=1)$ is specific to object's ones.
%$$
%y_i|\{x_j\}_{j\in V_n},\{y_j\}_{j\in V_n,j\neq i}
%\;\sim\;\left\{
%\begin{array}{ll}
%p_{+1}(y_i)& \hbox{if }x_i=1\,,\\
%p_{-1}(y_i) & \hbox{if }x_i=-1\,.
%\end{array}
%\right.
%$$
%Thus, pixels are conditionally independent given the allocation variables and the global parameters.
%A typical example would be the mixture of Gaussian ones
% $y_i|\{x_i\}_{i\in V_n},\{y_i\}_{i\in V_n}$
For example $p(y_i|x_i=-1)$ and $p(y_i|x_i=1)$ could be Gaussian distributions with parameters $(\mu_{-1},\sigma^2_{-1})$ and $(\mu_{1},\sigma^2_{1})$ respectively (see e.g. \cite{Moores2015application}).
The allocation variables $\{x_i\}_{i\in V_n}$ are given an Ising-type prior $p(\{x_i\}_{i\in V_n})\propto\exp(\lambda\sum_{(i,j)\in E_n}x_ix_j)$ to induce positive correlation among neighboring $x_i$'s.
It is easy to see that the resulting posterior distribution $p(\{x_i\}_{i\in V_n}|\{y_i\}_{i\in V_n})$ follows an Ising model as in \eqref{eq:ising_model} with biasing terms $\alpha_i$ given by $\frac{1}{2}\log\left(\frac{p(y_i|x_i=1)}{p(y_i|x_i=-1)}\right)$.

In Section \ref{sec:sampling_ising} we considered $n\times n$ grids, with $n$ ranging from $20$ to $1000$, and five levels of ``concentration'' for the target distribution.
The latter are obtained by considering an increasing spatial correlation ($\lambda=0,0.5,1,1,1$) and an increasingly informative external field $\{\alpha_i\}_{i\in V_n}$.
\begin{figure}[h!]
\includegraphics[width=\linewidth]{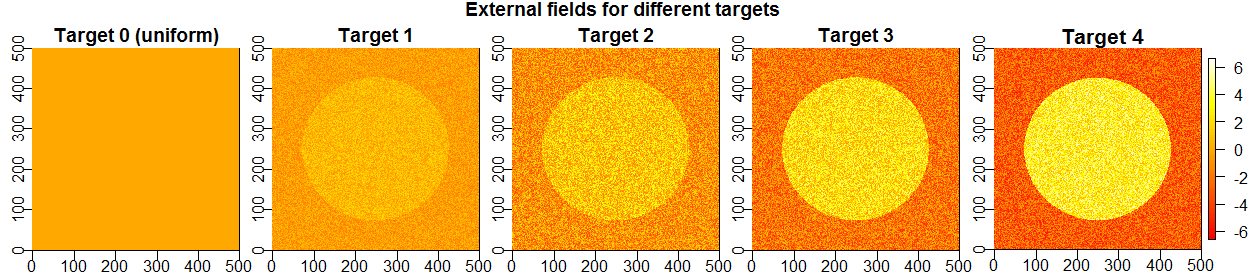}
\caption{External fileds $\{\alpha_i\}_{i\in V_n}$ for the 5 different targets considered in Section \ref{sec:sampling_ising}. Here $n=500$.}\label{fig:ising_fields}
\end{figure}
The external fields are illustrated in Figure \ref{fig:ising_fields}.
The object is at the center of the grid with a circular shape.
If $i$ is an object pixel we set $\alpha_i=\mu+Z_i$, while if it a background pixel we set $\alpha_i=-\mu+Z_i$, where $Z_i$ are iid $Unif(-\sigma,\sigma)$ noise terms.
For Target 0 up to Target 4 we considered $\mu=0,0.5,1,2,3$ and $\sigma=0,1.5,3,3,3$.
Note that Target 0 coincides with a uniform distribution on $\{-1,1\}^{n^2}$, i.e.\ $n^2$ i.i.d.\ Bernoulli random variables, where all schemes collapse to the same tranistion kernel.
Figures \ref{fig:traceplots_Ising_all} report traceplots, acceptance rates, effective sample sizes per unit time and autocorrelation functions for the five MCMC schemes for Targets 1-4. 
\begin{figure}[h!]
\includegraphics[width=\linewidth]{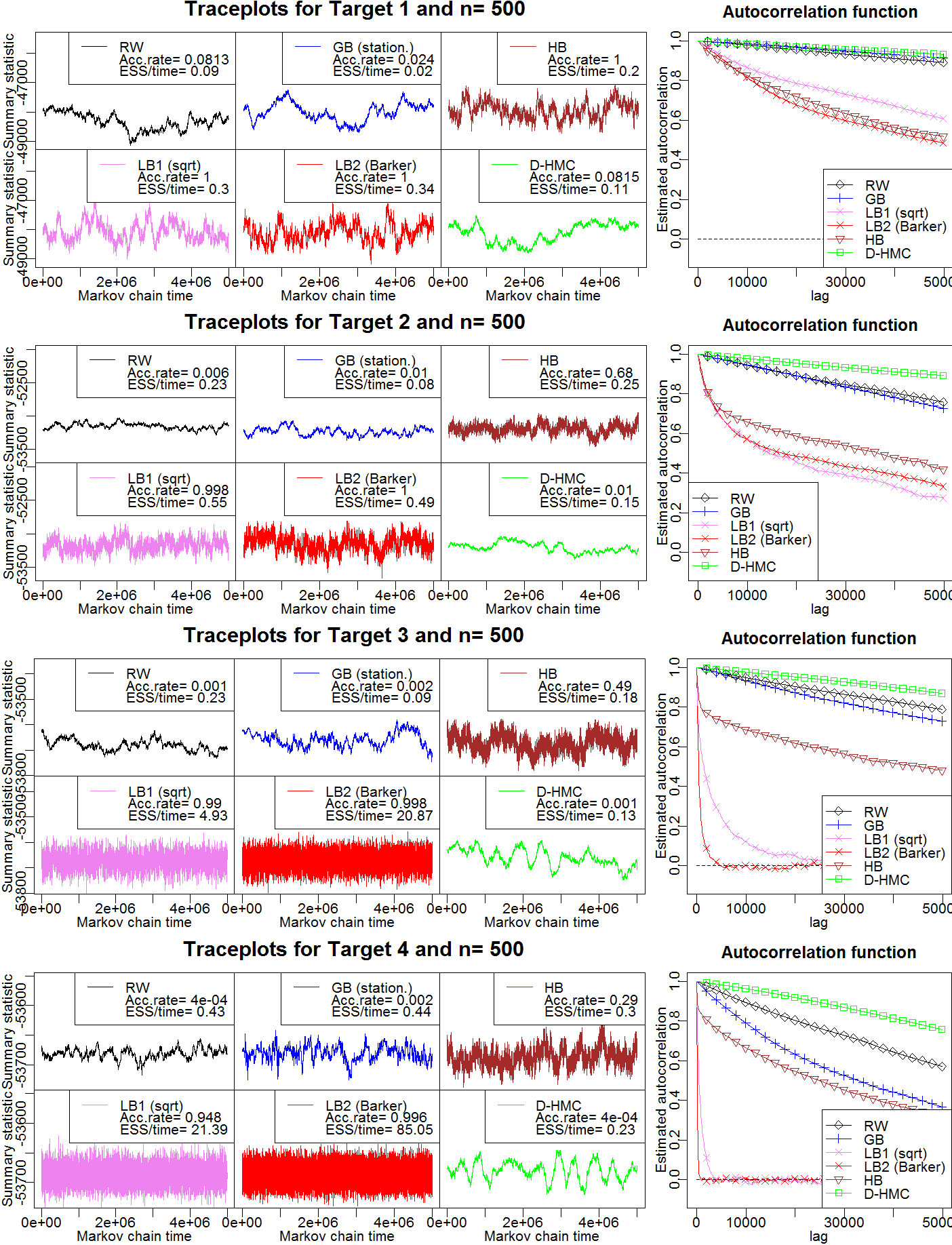}
\caption{Performances for the six MCMC under comparison in Section \ref{sec:sampling_ising} with target distributions described in Section \ref{supp:Ising}.
For D-HMC we are plotting the whole trajectory, including the path during the integration period.
Left: traceplots of a summary statistic (Hamming distance from fixed permutation). 
Right: estimated autocorrelation functions.
}\label{fig:traceplots_Ising_all}
\end{figure}

\newpage
\section{Supplement for the Record Linkage application}\label{appendix:record_linkage}
In this Supplement we describe in more details the Bayesian model for record linkage used in Section \ref{sec:record_linkage}.
%The model is fairly standard, with
The prior distribution follows mostly \cite{Miller2015} and \cite{Zanella2016NIPS}, while the likelihood distribution is similar to, e.g.\, \cite{CopasHilton1990} or \cite{SteortsHallFienberg2016}.
The parameter of interest is the matching $\textbf{M}=(M_1,\dots,M_{n_1})$, while the data are $\x=(x_1,\dots,x_{n_1})$ and $\y=(y_1,\dots,y_{n_2})$.
\subsection{Prior distribution on the matching structure}
We first specify a prior distribution for $\textbf{M}$, which is the unknown matching between $\{1,\dots,n_1\}$ and $\{1,\dots,n_2\}$.
First, we assume the prior distribution of $\textbf{M}$ to be invariant with respect to permutations of $\{1,\dots,n_1\}$ and $\{1,\dots,n_2\}$.
This is a standard a priori assumption which is appropriate when the ordering of $\x$ and $\y$ is expected to carry no specific meaning.
We then assume that the total number of entities in $\textbf{M}$ %(i.e.\ number of matched pairs $N_m$, plus number of singletons in the first group $n_1-N_m$, plus the number of singletons in the second group $n_2-N_m$).
(i.e.\ number of matched pairs plus number of singletons) follows a Poisson distribution a priori with unknown mean $\lambda$.
Given the number of entities, each of them either generates two records (one in $\x$ and one in $\y$) with probability $p_{match}$ or otherwise generates a single record (which is assigned either to $\x$ or to $\y$ with equal probability).
Both $\lambda$ and $p_{match}$ are treated as unknown hyperparameters.
Such prior distribution for $\textbf{M}$ can be seen as a simple, bipartite version of the more general priors for Bayesian record linkage discussed in \cite{Zanella2016NIPS}.
Note that also the sizes of $\x$ and $\y$, which we denote by $N_x$ and $N_y$ respectively, are random objects a priori to be conditioned on the observed values $n_x$ and $n_y$.
The resulting prior distribution for $\textbf{M}$ given $\lambda$ and $p_{match}$ is
\begin{align}\label{eq:prior_for_M}
P(\textbf{M}|\lambda,p_{match})&=
\frac{e^{-\lambda}\lambda^{N_x+N_y-N_m}}{N_x!N_y!}\left(\frac{1-p_{match}}{2}\right)^{N_x+N_y-2N_m}p_{match}^{N_m}\,,
\end{align}
where  $N_m$ denotes the number of matches in $\textbf{M}$.
\begin{proof}[Derivation of \eqref{eq:prior_for_M}]
First decompose the distribution of $\textbf{M}$ as
\begin{align*}
P(\textbf{M}|\lambda,p_{match})&=
P(\textbf{M}|N_m,N_x,N_y)
P(N_m,N_x,N_y|\lambda,p_{match})\,.
\end{align*}
By the exchangeability assumption, $P(\textbf{M}|N_m,N_x,N_y)$ is simply one over the number or partial matchings between $N_x$ and $N_y$ indices with $N_m$ matched pairs, i.e. 
$$
P(\textbf{M}|N_m,N_x,N_y)
=
{N_x\choose N_m}^{-1}{N_y\choose N_m}^{-1}\frac{1}{N_m!}
=
\frac{N_m!(N_x-N_m)!(N_y-N_m)!}{N_x!N_y!}\,.
$$ 
The expression for $P(N_m,N_x,N_y|\lambda,p_{match})$ can be easily computed as
$$
P(N_m,N_x,N_y|\lambda,p_{match})
=
\frac{e^{-\lambda}\lambda^{N_x+N_y-N_m} \left(\frac{1-p_{match}}{2}\right)^{N_x+N_y-2N_m}p_{match}^{N_m}}{(N_x-N_m)!(N_y-N_m)!N_m!}\,.
$$ 
noting that, given $\lambda$ and $p_{match}$, the three random variables $N_m$, $N_x-N_m$ and $N_y-N_m$ follow independent Poisson distributions with intensities $\lambda p_{match}$ for the first and $\lambda \left(\frac{1-p_{match}}{2}\right)$ for the last two.
Combining the previous expressions we obtain \eqref{eq:prior_for_M}.
\end{proof}
For the purpose of the simulation studies in Section \ref{sec:record_linkage} %of \cite{Zanella2017} 
we put independent, weakly informative priors on $p_{match}$ and $\lambda$, such as $\lambda\sim Unif([N_1\vee N_2,N_1+N_2])$ and $p_{match}\sim Unif([0,1])$.

\subsection{Likelihood distribution}
The distribution of $(\x,\y)|\textbf{M}$ follows a discrete spike and slab model, also known as hit and miss model in the record linkage literature \citep{CopasHilton1990}.
The observed data consist of two lists $\x=(x_i)_{i=1}^{n_x}$ and $\y=(y_j)_{j=1}^{n_y}$, where each element $x_i$ and $y_j$ contains $\ell$ fields, i.e.\ $x_i=(x_{is})_{s=1}^\ell$ and $y_j=(y_{js})_{s=1}^\ell$.
Each field $s\in\{1,\dots,\ell\}$ has a specific number of categories $m_s$ and density vector $\btheta_s=(\theta_{sp})_{p=1}^{m_s}$.
In real data applications, we assume $\btheta_s$ to be known and estimate it with the empirical distribution of the observed values of the $s$-th field.
This is a standard empirical bayes procedure typically used for this kind of models (see e.g.\ \cite{Zanella2016NIPS}).

%The observations for each field $s\in\{1,\dots,\ell\}$ are generated independently as follows.
%
%Observations for different fields $s\in\{1,\dots,\ell\}$ and for datapoints $x_i$ or $y_j$  that are not matched together are generated independently.

Entities are assumed to be conditionally independent given the matching $\textbf{M}$.
%Given $\textbf{M}$, datapoints that are not matched together (i.e.\ that do not represent the same entity) are generated independently.
If a datapoint, say $x_i$, is a singleton, then each field value $x_{is}$ is drawn from $\btheta_{s}$ %independently of the rest
\begin{align*}
\qquad\qquad
&&x_{is}\sim\btheta_{s}&&s\in\{1,\dots,\ell\}\,,
\end{align*}
and similarly for singletons $y_j$.
If instead $x_i$ and  $y_j$ are matched to each other, first a common value $v_s$ is drawn from $\btheta_{s}$ for each field $s$ and then, given $v_s$, $x_{is}$ and $y_{js}$ are independently sampled from a mixture of  $\btheta_{s}$ and a spike in $v_s$
\begin{align*}
&&x_{is},y_{js}|v_s\stackrel{iid}\sim \beta \btheta_{s}+(1-\beta)\delta_{v_s}&&s\in\{1,\dots,\ell\}\,,
\end{align*}
where $\beta$ is a distortion probability in $(0,1)$.
%In our context we assume $\beta$ to be known and we set it to $0.001$ which is a realistic value for the Italy dataset (see e.g.\ discussion on distortion parameter in \cite{Steorts2015EB}).
%Note that in principle one would like to include $\beta$ in the Bayesian model as an unknown hyperparameter.
%However, the latter is a hard parameter to estimate and previous Bayesian works in this context required very strong prior information (see e.g.\ the discussion in \cite{Steorts2015EB}).
In principle one would like to include the distortion probability $\beta$ in the Bayesian model as an unknown hyperparameter.
However, the latter is a difficult parameter to estimate in this contexts and previous Bayesian works in this context required very strong prior information (see e.g.\ the discussion in \cite{Steorts2015EB}).
For simplicity, in our context we assume $\beta$ to be known and we set it to $0.001$ which is a realistic value for the Italy dataset according to the discussion in \cite{Steorts2015EB}.
We ran simulations for different values of $\beta$ (e.g. $0.01$, $0.005$, $0.002$) and the results regarding relative efficiency of the algorithms under consideration were consistent with the ones reported in Section \ref{sec:record_linkage}. % in terms of did not change significantly.
Finally, we obtain the following expression for the likelihood function
%\begin{multline}\label{eq:record_linkage_likelihood}
%P(\x,\y|\textbf{M},(\btheta_s)_{s=1}^\ell,\beta)
%=
%\left(\prod_{s=1}^\ell
%\left(\prod_{i=1}^{n_x}\theta_{\ell x_{is}}\right)
%\left(\prod_{j=1}^{n_y}\theta_{\ell y_{js}}\right)\right)\\
%\prod_{s=1}^\ell
%\prod_{i\sim j}
%\left(
%\beta_\ell\left(2-\beta_\ell\right)
%+
%\frac{(1-\beta_\ell)^2}{\theta_{\ell x_{i\ell}}}
%\mathbbm{1}(x_{i\ell}=y_{j\ell})
%\right)
%\,,
%\end{multline}
%where $\prod_{i\sim j}$ denotes the product over all matched couples of indices $(i,j)$.
 \begin{multline}\label{eq:record_linkage_likelihood}
%P(\x,\y|\textbf{M},(\btheta_s)_{s=1}^\ell,\beta)
P(\x,\y|\textbf{M})
=
\left(\prod_{s=1}^\ell
\left(\prod_{i=1}^{n_x}\theta_{s x_{is}}\right)
\left(\prod_{j=1}^{n_y}\theta_{s y_{js}}\right)\right)\\
\prod_{s=1}^\ell
\prod_{i\,:\,M_i>0}
\left(
\beta\left(2-\beta\right)
+
\frac{(1-\beta)^2}{\theta_{s x_{is}}}
\mathbbm{1}(x_{is}=y_{M_is})
\right)
\,,
\end{multline}
where $\{i\,:\,M_i>0\}$ is the set of indices $i\in\{1,\dots,n_x\}$ that are matched to some $j\in\{1,\dots,n_y\}$.
\begin{proof}[Derivation of \eqref{eq:record_linkage_likelihood}]
The conditional independence over fields and datapoints that are not matched implies that $P(\x,\y|\textbf{M})$ can be factorized as
%\begin{align*}
%\prod_{s=1}^\ell\left(
%\left(
%\prod_{i\,:\,M_i=0}P(x_{is}|\btheta_s,\beta)
%\right)
%\left(%\prod_{j\notin \{M_1,\dots,M_{n_x}\}}
%\prod_{j\,:\,j\neq M_i\forall i}P(y_{js}|\btheta_s)\right)
%\left(\prod_{i\,:\,M_i>0}P(x_{js},y_{M_is}|\btheta_s)\right)
%\right)
%\,.
%\end{align*}
\begin{align*}
\prod_{s=1}^\ell\left(
\left(
\prod_{i=1}^{n_x}\theta_{sx_{is}}
\right)
\left(%\prod_{j\notin \{M_1,\dots,M_{n_x}\}}
\prod_{j=1}^{n_y}\theta_{sy_{js}}
\right)
\left(\prod_{i\,:\,M_i>0}\frac{P(x_{js},y_{M_is}|\textbf{M})}{\theta_{sx_{is}}\theta_{sy_{M_is}}}\right)
\right)
\,.
\end{align*}
The term $\frac{P(x_{js},y_{M_is}|\textbf{M})}{\theta_{sx_{is}}\theta_{sy_{M_is}}}$ equals
\begin{align*}
\sum_{p=1}^{m_s}
\theta_{s p}
\frac{\left(\beta\theta_{s x_{is}}+(1-\beta)\mathbbm{1}(x_{is}=p)\right)}{\theta_{sx_{is}}}
\frac{\left(\beta\theta_{s y_{js}}+(1-\beta)\mathbbm{1}(y_{js}=p)\right)}{\theta_{sy_{M_is}}}\,.
\end{align*}
%\begin{align*}
%P(x_{is},y_{M_is}|\btheta_s)=\sum_{p=1}^{m_s}
%\theta_{s p}
%\left(\beta\theta_{s x_{is}}+(1-\beta)\mathbbm{1}(x_{is}=p)\right)
%\left(\beta\theta_{s y_{js}}+(1-\beta)\mathbbm{1}(y_{js}=p)\right)\,.
%\end{align*}
With simple calculations one can see that, if $x_{is}=y_{M_is}$ the expression above equals
$
\beta(2-\beta)+
\frac{(1-\beta)^2}{\theta_{s x_{is}}}
$,
while if $x_{is}\neq y_{M_is}$ it equals
$\beta\left(2-\beta\right)$.
Combining the latter equations we obtain \eqref{eq:record_linkage_likelihood}.
\end{proof}

\subsection{Metropolis-within-Gibbs sampler}\label{appendix:Metropolis_within_Gibbs}
We use a Metropolis-within-Gibbs scheme to sample from the posterior distribution of $(\textbf{M},\lambda,p_{match})|(\x,\y)$.
%The unknown parameters to be sampled from the posterior are $(\textbf{M},\lambda,p_{match})$.
The two hyperparameters $\lambda$ and $p_{match}$ are conditionally independent given $(\x,\y,\textbf{M})$, with full conditionals given by% distribution of $p_{match}$ and $\lambda$ are
\begin{align}
p_{match}|\,\x,\y,\textbf{M}\sim&\;\hbox{Beta}(1+n_x+n_y-2N_m,1+N_m)\,,\label{eq:full_cond_p}\\
\lambda|\,\x,\y,\textbf{M}\sim&\; \hbox{Gamma}_{[\min\{n_x,n_y\},n_x+n_y]}(1+n_x+n_y-N_m,1)\,,\label{eq:full_cond_lambda}%\1(\lambda\in(N_1\vee N_2,N_1+N_2))
\end{align}
where $N_m$ is the numer of matched pairs in $\textbf{M}$ and $\hbox{Gamma}_{[a,b]}$ denotes a Gamma distribution truncated on the interval $[a,b]$. 
%The full conditional distribution of $\lambda$ given $\textbf{x},\textbf{M},p_{match}$ is a $\hbox{Gamma}(1+n_x+n_y-N_m,1)$ distribution truncated on $\lambda\in[\min\{n_x,n_y\},n_x+n_y]$.
%$$
%P(\lambda|\textbf{x},\textbf{M},p_{match})
%\propto
%P(\textbf{M}|\lambda,p_{match})
%P(\lambda)
%\propto
%e^{-\lambda}\lambda^KP(\lambda)
%$$
%$$
%P(p_{match}|\textbf{x},\textbf{M},\lambda)
%\propto
%P(\textbf{M}|\lambda,p_{match})
%P(p_{match})
%=
%p_{match}^{S_b}\left(1-p_{match}\right)^{S_1+S_2}P(p_{match})
%$$
%We use weakly informative priors for both hyperparameters.
%In particular we assume $p_{match}\sim Unif(0,1)$ and $\lambda\sim Unif([N_1\vee N_2,N_1+N_2])$ a priori.
%%on $\lambda$, such as $\lambda\sim Unif([N_1\vee N_2,N_1+N_2])$ or $\lambda\sim Gamma(1,\frac{N_1+N_2+N_1\vee N_2}{2})$. Both have prior expectation of $\frac{N_1+N_2+N_1\vee N_2}{2}$ and standard deviation of the \\
%It follows that $p_{match}|\textbf{x},\textbf{M},\lambda\sim Beta(1+S_1+S_2,1+S_b)$ and $\lambda|\textbf{x},\textbf{M},p_{match}\sim Gamma(K+1,1)\1(\lambda\in(N_1\vee N_2,N_1+N_2))$.
%%Under the $Unif([N_1\vee N_2,N_1+N_2])$ prior we have that $
%%We put a $p_{match}\sim Unif(0,1)$ prior on $$.
The full conditional distribution of $\textbf{M}$ is proportional to
\begin{multline}\label{eq:full_cond_M}
P(\textbf{M}|\x,\y,\lambda,p_{match})\propto\\
\prod_{i\,:\,M_i>0}
\left(\frac{4p_{match}}{\lambda(1-p_{match})^2}
\prod_{s=1}^\ell
\left(
\beta\left(2-\beta\right)
+
\frac{(1-\beta)^2}{\theta_{s x_{is}}}
\mathbbm{1}(x_{is}=y_{M_is})
\right)\right)
\,.
\end{multline}
The Metropolis-within-Gibbs scheme alternates Gibbs updates for $p_{match}$ and $\lambda$ according to \eqref{eq:full_cond_p} and \eqref{eq:full_cond_lambda} with Metropolis updates for $\textbf{M}$ according to \eqref{eq:full_cond_M}.
The challenging and computationally intense step is updating $\textbf{M}$ in an effective way and that's where we compare different sampling schemes (see Section \ref{sec:record_linkage}).
The four sampling schemes we compare %in Section \ref{sec:record_linkage}
 (denoted by RW, GB, LB and HB) are all based on the same base kernel $K(\textbf{M},\textbf{M}^\prime)$, which proceed as follows.
First a couple $(i,j)$ is sampled uniformly at random from $\{1,\dots,n_x\}\times\{1,\dots,n_y\}$.
Then, given $\textbf{M}$ and $(i,j)$, one of the following moves is performed:
\begin{itemize}[topsep=0pt,itemsep=-1ex,noitemsep]
\item Add move: if $M_i=0$ and $M^{-1}_j=0$, set $M_i=j$;
\item Delete move: if $M_i=j$, set $M_i=0$;
\item Single switch move (I): if $M_i=0$ and $M^{-1}_j=i'$ for some $i'\in\{1,\dots, n_x\}$, set $M_i=j$ and $M_{i'}=0$;
\item Single switch move (II): if $M_i=j'$ for some $j'\in\{1,\dots, n_y\}$ and $M^{-1}_j=0$, set $M_i=j$;
\item Double switch move: if $M_i=j'$ for some $j'\in\{1,\dots, n_y\}$ and $M^{-1}_j=i$ for some $i'\in\{1,\dots, n_x\}$, set $M_i=j$ and $M_{i'}=j'$.
\end{itemize}
Here $M^{-1}_j$ is defined as $M^{-1}_j=i$ if $i$ is currently matched with $j$ and $M^{-1}_j=0$ if no index is matched to $j$.

\subsection{Blocking implementation}\label{appendix:blocking}
In this context, sampling from GB, LB and HB has a computational cost that grows with $n_x$ and $n_y$.
When the latter terms are too large, it is more efficient to apply informed schemes to sub-blocks of indices (see discussion in Section \ref{sec:blocking}).
In Section \ref{sec:record_linkage} we used the following block-wise implementation:
%More specifically, the blocking implementation proceeds as follows:
\begin{enumerate}[topsep=0pt,itemsep=-1ex,noitemsep]
\item choose a subset of indices $I\subset\{1,\dots,n_x\}$ and a subset $J\subset\{1,\dots,n_y\}$ (details below)
\item remove the indices $i\in I$ that are matched with some $j\notin J$ and, similarly, remove the indices $j\in J$ that are matched with some $i\notin I$ 
\item update the components $(M_i)_{i\in I}$ by performing $k$ iterations of the MCMC kernel under consideration restricting the proposal of couples $(i,j)$ to $I\times J\subseteq \{1,\dots,n_x\}\times\{1,\dots,n_y\}$.
\end{enumerate}
Step 2 is needed to avoid moves that would involve indices outside $I\times J$. It is easy to check that, if the MCMC used on the $I\times J$ subspace is invariant with respect to $\prod_{i\in I}w_{iM_i}$, then the whole chain is invariant with respect to  $\prod_{i=1}^{n_x}w_{iM_i}$.
There are many valid ways to choose the subsets $I$ and $J$ in Step 1 (both randomly and deterministically).
In our implementation in Section \ref{sec:record_linkage} we set the size of $I$ and $J$ to 300.
In order to favour couples having a non-negligible probability of being matched we alternated a uniformly at random selection of $I$ and $J$ (which ensures irreducibility) with the following procedure:
%First we sample 3 features $\{s_1,s_2,s_3\}\subseteq\{1,\dots,\ell\}$ uniformly at random and then one value for each feature $p_t\in\{1,\dots,m_{s_t}\}$ for $t=1,2,3$ proportional to the data empirical distribution.
sample an index $i_0\in\{1,\dots,n_x\}$ and three features $\{s_1,s_2,s_3\}\subseteq\{1,\dots,\ell\}$ uniformly at random; define $I$ and $J$ as the set of $i$'s and $j$'s such that $x_i$ and $y_j$ respectively agree with $x_{i_0}$ on the three selected features;
if the size of $I$ or $J$ exceeds 300, reduce it to 300 with uniform at random subsampling.
%First we sample an index $i_0\in\{1,\dots,n_x\}$ and a feature $s\in\{1,\dots,\ell\}$ u.a.r., then we set $I$ and $J$ as the set of $i$'s and $j$'s such that $x_i$ and $y_j$ respectively agree with $x_{i_0}$ on the $s$ feature.
%We repeat the procedure 
We claim no optimality for the latter procedure and defer the detailed study of optimal block-wise implementations to further work (see Section \ref{sec:discussion}).
\end{document}